\pdfoutput=1
\documentclass{amsart}

\title[Multifractal analysis of intermingled basins and blowout bifurcations]{Multifractal analysis of intermingled basins 
and blowout bifurcations in a parametetric family of skew product maps}
\author{Fatemeh Helen Ghane}
\address{Department of Mathematics, Ferdowsi University of Mashhad, Iran}
\email{htina@um.ac.ir}
\author{Marc Kesseb\"ohmer}
\address{Institute for Dynamical Systems, Faculty 3 -- Mathematics  and Computer Science, University of  Bremen, Germany}
\email{mhk@uni-bremen.de}
\date{\today}
\subjclass{37C70,  37C40,  37H15, 37C45}
\keywords{Invariant Graph, Chaotic Milnor Attractor, Riddled Basin, Intermingled Basin, Blowout Bifurcation, Stability Index, Multifractal Analysis}
 
\usepackage{tikz}
\usepackage{mathtools}
\usetikzlibrary{patterns}
\usepackage{pgfplots}
\usepgfplotslibrary{colormaps,fillbetween}
\pgfplotsset{compat=1.15}
\usepackage{mathrsfs}
\usetikzlibrary{arrows}

\usepackage[T1]{fontenc}
\usepackage{geometry}
\usepackage{xcolor}
\usepackage{listings}
\usepackage{fancyvrb}
\usepackage{amsmath, amsthm,amscd, amsfonts, amssymb, graphicx, color,float}
\usepackage{subfig}
 
\usepackage{color} %red, green, blue, yellow, cyan, magenta, black, white
\usepackage{fancyhdr}
\usepackage{amsthm}
\usepackage{algorithm,algorithmic}

\newtheorem{theorem}{Theorem}
\newtheorem{proposition}{Proposition}[section]

\newtheorem{corollary}[proposition]{Corollary}

\newtheorem{lemma}[proposition]{Lemma}

\theoremstyle{definition}
\newtheorem{definition}[proposition]{Definition}
\newtheorem*{condition*}{Condition}
\theoremstyle{remark}
\newtheorem{remark}[proposition]{Remark}

\numberwithin{equation}{section}
\renewcommand{\d}{\;\mathrm{d}}
\newcommand{\e}{\mathrm{e}}
\let\lvert=|\let\rvert=|

\usepackage[style=authoryear,backend=biber,style=alphabetic,  backref=false, giveninits=true,citestyle=alphabetic, natbib=false, url=false, isbn=false,doi=true,eprint=false,maxbibnames=99]{biblatex}
\DeclareFieldFormat[article]{title}{{#1}}
\DeclareFieldFormat[unpublished]{title}{{#1}}
\DeclareFieldFormat[proceedings]{title}{{#1}}
\DeclareFieldFormat[incollection]{title}{{#1}}
\DeclareFieldFormat[unpublished]{title}{{\em #1}}
\renewbibmacro{in:}{}
\DeclareFieldFormat[article]{pages}{#1}

\begin{document}

\begin{abstract}
In this paper we study a two-parameter family of planar maps characterized by two distinct invariant subspaces. The model reveals the existence of two chaotic attractors  
within these subspaces.  We identify parameter values at which these attractors either exhibits a locally riddled basin of attraction or transitions into a chaotic saddle.  In particular, we demonstrate that, for an open region in the parameter plane, their basins are intermingled. It is shown that a fractal boundary curve separates the basins of attraction of these two chaotic attractors, providing a detailed characterization of the riddled basin structure.   Additionally, we show that the model undergoes a blowout bifurcation.  An estimation of the stability index is examined using thermodynamic formalism. We also perform a multifractal analysis of the level sets of the stability index.  
\end{abstract}
 
\maketitle
\section{Introduction}
Recent interest in the global dynamics of systems with multiple attractors has emphasized the often complex structure of their basins of attraction. The phenomenon of multiple attractors for which several attractors coexist has been explored e.\,g.\@ in  \cite{buescu2012exotic, ott2002chaos, daza2016basin, dudkowski2016hidden}. Our focus is on systems with multiple attractors that display densely interwoven basins of attraction, a phenomenon known as \emph{riddling}.  This implies that for any initial condition within the basin of one attractor there are arbitrarily close initial conditions that converge to  another attractor.  Ott et al. \cite{ott1993scaling} introduced nonlinear dynamical systems with simple symmetries that exhibit riddled basins. Conditions for the occurrence of riddled basins were further defined by Alexander et al. \cite{alexander1992riddled} and subsequently generalized by Ashwin et al.  \cite{ashwin1996attractor}. 
Intermingled basins describe a scenario in which multiple attractors share overlapping regions within the phase space. This overlap often leads to intricate interactions among the attractors, producing a diverse and complex range of dynamical behaviors. Such phenomena can exhibit chaos, characterized by the system's sensitive dependence on initial conditions. Intermingled basins are commonly associated with complex dynamics, including fractal basin boundaries, where the basin structure is highly intricate and challenging to identify.   
A detailed picture of multiple attractors with riddled or intermingled basins is available through works of several authors \cite{alexander1992riddled, ashwin1996attractor, lai2005basins, roslan2016local, schultz2017potentials, saha2018riddled, rabiee2022occurrence}.

 Understanding intermingled basins is crucial in diverse applications, such as coupled nonlinear electronic circuits,  \cite{ashwin1994bubbling,heagy1994experimental},
  a forced double-well Duffing oscillator  \cite{sommerer1993physical, ott1993scaling, ott1994transition},
ecological population models \cite{cazelles2001dynamics, viana2009riddled, karimi2020analysis},
learning dynamical systems \cite{nakajima1996riddled}, 
 and engineering systems, where stability and state transitions can have far-reaching impacts. 
 The paper \cite{kim2018multistability} examines multistability in power-grid systems, highlighting how variations in basins of attraction influence complex dynamics. It delves into the role of intermingled basins and their interactions, offering insights into stability and transitions in power-grid networks. 

 Riddling is commonly observed in skew product systems due to the asymmetry in how variables interact \cite{alexander1992riddled, bonifant2008schwarzian, cazelles2001dynamics}. 
 These systems are random dynamical systems influenced by a deterministic external factor. The theory of these systems has been developed to provide a framework for modeling dynamics subjected to random perturbations. 

 In this context, the local dynamic stability of a chaotic attractor can be assessed using Lyapunov exponents. When the largest Lyapunov exponent is negative, a set of positive measure exists that is asymptotically attracted to the attractor  \cite{alexander1992riddled, ashwin1996attractor}. 
For most chaotic attractors, ergodic measures are not unique and can include Dirac measures supported by periodic orbits.  

When two chaotic attractors reside in separate invariant subspaces, the system forms a complex fractal boundary between the initial conditions that lead to each attractor. In intermingled basins, small changes in initial conditions can cause the system to switch between attractors, making the resulting trajectory in phase space highly unpredictable. 

The first example of maps with intermingled basins was provided by Kan \cite{kan1994open}, who studied a partially hyperbolic endomorphism on a surface. This system featured a boundary exhibiting two intermingled hyperbolic physical measures, highlighting the complex dynamics and overlapping attractor basins. 
Keller \cite{keller2017stability} investigated the phenomenon of intermingled basins within a skew product dynamical system defined on a square. The system featured a piecewise expanding Markov base map, along with a fiber map exhibiting a negative Schwarzian derivative. 

In \cite{podvigina2011local}, the concept of a stability index for a basin was introduced. This index quantifies the degree to which a basin is riddled at a given point. Essentially, the stability index measures the local chaotic behavior of the basin at a specific point, providing a means to assess the complexity or chaotic nature of the dynamics in that region of phase space. 

As the study of riddled and intermingled basins advanced, multifractal analysis became a popular method for investigating the complex structure of basin boundaries. Researchers like Keller \cite{keller2014stability, keller2017stability}, Walkden and Withers \cite{walkden2017stability} applied multifractal techniques to quantify the irregularities in basin boundaries. This analysis led to the realization that basin boundaries exhibit a range of scaling behaviors and fractal structures. 
By applying multifractal analysis, we can rigorously quantify different dynamical behaviors in terms of the fractal dimensions of dynamically defined subsets \cite{schmeling1999completeness, pesin1997multifractal, kessebohmer2007multifractal, kessebohmer2008fractal, jaerisch2021multifractal, jaerisch2011regularity}.

In this article we analyze the behavior of a two parameter family  $\mathcal{F}$ of planar systems $F_{a,b}$ acting on the unit square.
In our setting, each system $F_{a,b}$ exhibits two invariant subspaces $\Phi_i$ having chaotic attractors $A_i$, $i=0,1$.
We demonstrate the conditions for the emergence of a locally riddled basin and chaotic saddle.
Here, the two parameter  plane maps $F_{a,b}$ are skew product maps, where the base map $f$  is a piecewise expanding and piecewise $C^{1+\alpha}$-Hölder mixing Markov map with two branches. The fiber maps are diffeomorphisms defined on the unit interval $\mathbb{I}$ such that they fix the two endpoints, 0 and 1.
Therefore, we have two invariant sets
\begin{equation}\label{subspaces}
\Phi_0\coloneqq \mathbb{I} \times \{0\}, \quad \Phi_{1}\coloneqq\mathbb{I} \times \{1\}.
\end{equation}
This invariant sets play the role of the $F_{a,b}$-invariant manifolds with chaotic dynamics inside.
One objective is to characterize the different possible dynamics by varying the parameters $a, b$. These parameters vary the transverse dynamics without changing the dynamics on the invariant sets $\Phi_0$ and $\Phi_1$.
We show that for certain values of $a$ and $b$ within an open region of the $ab$-plane, the skew product map $F_{a,b}$ has two chaotic Milnor attractors in the invariant sets $\Phi_0$ and $\Phi_1$. The system's qualitative dynamics depend on initial conditions, and the attractors exhibit an intricate, intermingled basin structure.
This phenomenon produces an unpredictability qualitatively greater than the traditional sensitive dependence on initial conditions within a single chaotic attractor.
By  defining a fractal boundary between the attractors' basins, we investigate the dynamics of riddled basins and blowout bifurcations thoroughly. 
The system's dynamics are characterized by two Lyapunov exponents. The first, the parallel Lyapunov exponent, describes evolution within the invariant subspaces and must be positive for riddled basins to emerge. The second, the normal Lyapunov exponent, characterizes evolution transverse to these subspaces \cite{ashwin1996attractor, cazelles2001dynamics, viana2009riddled}.
We further investigate blowout bifurcations of chaotic attractors within invariant subspaces, analyzing their occurrence in detail. First, we estimate the parameter range $(a,b)$ where the attractors $A_0$ and $A_1$ exhibit intermingled basins or transition into chaotic saddles, providing a rigorous analysis of these complex behaviors. Additionally, we demonstrate that varying $a$ and $b$ induces a blowout bifurcation.

Using Keller’s criterion  \cite{keller2017stability}, we compute the stability indices for both chaotic attractors $A_0$ and $A_1$.
Keller \cite{keller2017stability} introduced a formalism that integrates the stability index and thermodynamic measures to comprehensively describe intermingled basins. In this article, we apply this formalism along with Ashwin's approach \cite{ashwin1996attractor} to our parametric maps, focusing on exploring different dynamical regimes and the bifurcations induced by parameter variations.    Additionally, we perform a multifractal analysis to determine the Hausdorff dimension of the level sets of the stability indices.  
In particular, we obtain three dynamic regimes which are typically separated by bifurcation curves where stability changes.
Dynamic regimes of our two-parameter map represent the qualitative behaviors of the system's trajectories in its state space as the parameters $a$ and $b$ are varied.  

Note that while many attracting invariant graphs are Milnor attractors, it is not a universal rule. The distinction lies in the measure-theoretic attraction criterion: a Milnor attractor requires attracting a set of positive Lebesgue measure, whereas an attracting invariant graph may only attract a small or even negligible portion of the phase space. Milnor attractors can coexist with other invariant sets, including repellers or saddle-type invariant structures.

This paper is organized as follows. Section \ref{section2} provides a detailed explanation of the key concepts and terminology used throughout the paper and outlines the main results. In Section \ref{section3}, we focus on analyzing the two-parameter family $F_{a,b}$. Using the approach of \cite{ashwin1996attractor}, we explore the emergence of locally riddled basins and chaotic saddles for certain parameter values within an open region of the $ab$-plane. In Section \ref{section4}, we generalize \cite[Theorem~1]{keller2017stability} and calculate the stability indices of the chaotic attractors $A_i$, $i=0,1$ for general Gibbs measures. 
In Section \ref{section5}, we perform a multifractal analysis of the Hausdorff dimension for the level sets of the stability index.

 \section{Preliminaries and main results }\label{section2}
In this section, we present the fundamental concepts and notations that form the basis of this paper.  

\subsection{ Attractors and riddled basins}
Let $M$ be a compact, connected, smooth Riemannian manifold, and let $m$ denote the normalized Lebesgue measure on $M$. We begin by recalling some classical definitions associated with attractors. 

Consider a continuous map $F : M \to M$  and a compact $F$-invariant ser  $A\subset M$  (i.\,e.\@ $F(A) = A$). Let $\omega(x)$ be the set of $\omega$-limit points of the orbit $\{F^{n}(x)\}_{n\geq0}$. 
The \emph{basin of attraction of} $A$, which we denote by $\mathcal{B}(A)$, is the set of points whose $\omega$-limit set is contained in $A$.
For non-empty $A$ the basin $\mathcal{B}(A)$ is always non-empty because it includes $A$. For $A$ to be an attractor, we require that $\mathcal{B}(A)$ is large in the appropriate sense.
Assume $A$ is compact invariant set. Then 
\begin{enumerate}
    \item $A$ is an \emph{asymptotically stable attractor} if it is Lyapunov stable and the basin of attraction $\mathcal{B}(A)$ contains a neighborhood of $A$.
\item   $A$ is a \emph{Milnor attractor} \cite{milnor1985concept}  if the basin of attraction $\mathcal{B}(A)$ has positive Lebesgue measure. To be more precise, we say that $A$ is a Milnor attractor if $\mathcal{B}(A)$ has non-zero Lebesgue measure and there is no compact proper subset $A^{\prime}$ of $A$ whose basin coincides with $\mathcal{B}(A)$ up to a set of measure zero.
\item $A$ is called an \emph{essential attractor}  \cite{melbourne1991example} if $A$ is a Milnor attractor and
\begin{equation*}
\lim_{\delta \to 0}\frac{m(B_\delta(A) \cap \mathcal{B}(A))}{m(B_\delta(A))}=1,
\end{equation*}
where $B_\delta(A)$ is a $\delta$-neighborhood of $A$ in $M$.
\item $A$ is a \emph{chaotic attractor} if $A$ is a transitive Milnor attractor and supports an ergodic measure $\mu$ that is not uniquely ergodic. In particular, at least one of the Lyapunov exponents (with respect to $\mu$) is positive.
\end{enumerate}
Some dynamical systems have chaotic attractors with densely intertwined basins of attraction, which we refer to as \emph{riddled basin}.
Riddled basins were introduced in 1992 by \cite{alexander1992riddled} as follows:
A basin $\mathcal{B}(A)$ of an attractor $A$ is riddled (with holes in a measure-theoretic sense) if and only if its  complement $\mathcal{B}(A)^c$  intersects every disk in a set of positive measure.

This concept has been  generalized to a  local version: 
A Milnor attractor $A$ has a \emph{locally riddled basin} if there exists a
neighborhood $U$ of $A$ such that, for all $x \in A$ and $\varepsilon > 0$
\begin{equation}\label{l-riddel}
  m\left(B_{\varepsilon}(x)  \cap \left(\bigcap_{n\geq 0}F^{-n}(U)\right)^c\right) >0.
\end{equation}

If there is another Milnor attractor $A^{\prime}$ such that $\mathcal{B}(A)^{c}$ in the definition of locally riddled basin may be replaced with $\mathcal{B}(A^{\prime})$, then we say that the basin of $A$ is riddled with the basin of $A^{\prime}$.
If $\mathcal{B}(A)$ and $\mathcal{B}(A^{\prime})$ are riddled with each other, we say that they are \emph{intermingled}.

The corresponding concepts can be defined for
repelling sets. An invariant transitive set $A$ is a
\emph{chaotic saddle} if there exists a neighborhood $U$ of $A$ such that $\mathcal{B}(A)\cap U \neq \emptyset$ but $m(\mathcal{B}(A)) = 0$.

\subsection{The Basic Model}
We study a family $\mathcal{F}$ of two-parameter family of skew product maps with parameters $(a,b)\in (0,1/2)^2$ of the form
\begin{align}\label{ss}
F_{a,b}: \mathbb{I}\times \mathbb{I} \to \mathbb{I}\times \mathbb{I}, \ F_{a,b}(x,y)\coloneqq(f(x),g_{a,b}(x,y)),
\end{align}
where $\mathbb{I}$ is the unit interval $[0,1]$, $f$ is an expanding Markov map, which  for the ease of exposition is given by $f(x)=2x \mod 1$, 
and
\begin{equation}\label{fiber1}
g_{a,b}(x,y)=
\begin{cases}
g_a(y)=y+ay(1-y)\quad &\text{ if } 0 \leq x<1/2\\
g_b(y)=y-by(1-y)  \quad &\text{ if } 1/2 \leq x \leq 1.\\
\end{cases}
\end{equation}

We set $g_{a,b,x}\coloneqq g_a$ if $x \in [0, 1/2)$, and $g_{a,b,x}\coloneqq g_b$ if $x \in [1/2 , 1]$. For the iterates $F_{a,b}^{n}$ of $F$ we adopt the usual notation  $ F_{a,b}^{n}(x,y ) = (f^{n}x , g_{a,b,x}^{n}(y))$  where $g^{n}_{a,b,x} = g_{a,b, f^{n-1}(x)} \circ \cdots \circ g_{a,b,x}.$  Hence  $g_{a,b, x}^{n+k}(y) = g_{a,b, f^k x}^n (g_{a,b,x}^k(y))$ and for $n = 1$ and $k=-1$, this includes the identity.
\begin{remark}\label{Holder}
    We present our arguments for the case that $f$ is doubling map as above. However, our arguments can be extended to the general case, when the base map $f$  is a piecewise expanding and piecewise $C^{1+\alpha}$-Hölder mixing Markov map with two branches. 

The function $(x,y) \mapsto \log dg_{a,b,x} (y) $ is $\alpha$-Hölder continuous on each set $I_i \times \mathbb{I}$, $i=0,1,$ where $I_i$ is a Markov interval of $f$.

\end{remark}
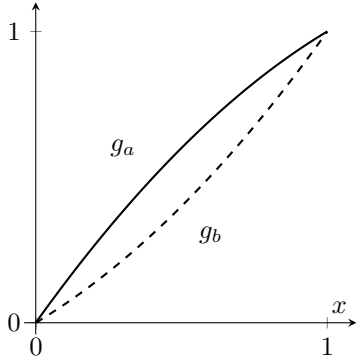
\begin{figure}[!ht]
\begin{center}
\begin{tikzpicture}
    \begin{axis}[
        axis lines=middle,
        width=6cm,
        height=6cm,
        xmin=-.04, xmax=1.1,
        ymin=-.04, ymax=1.1,
        xlabel={$x$},
        ylabel={},%$g_\bullet$},
        xtick={0.001, 1},
        ytick={0.001, 1},
        xticklabels={$0$, $1$},
        yticklabels={$0$, $1$},
        samples=100,
        domain=0:1,
        legend style={at={(0.95,0.05)}, anchor=south east}
    ]

    % Plot g_a(y) = y + (1/2)y(1-y) as solid line
    \addplot [thick, solid] 
    {x + 2/5 * x * (1 - x)};
  \node at (axis cs:0.3,0.6)  {$g_a$};

    % Plot g_b(y) = y - (1/2)y(1-y) as dashed line
    \addplot [thick, dashed] 
    {x - 2/5 * x * (1 - x)};
 \node at (axis cs:0.6,0.3)  {$g_b$};

    % Add legend
%    \legend{$g_a(y) = y + \frac{1}{2} y (1-y)$, $g_b(y) = y - \frac{1}{2} y (1-y)$}

    \end{axis}
\end{tikzpicture}
\end{center}
  \caption{This plot represents the fiber maps $g_a$ (solid) and $g_b$ (dashed) for $a=b=2/5$.}
 \label{fig:1}
   \end{figure}

 \subsection{Invariant Graphs and Lyapunov exponents}
 
Invariant graphs are essential objects in the study of skew product systems and are of significant interest. They often arise in systems where a base dynamics drives or interacts with a dependent (fiber) dynamics. 
A thorough understanding of invariant graphs in skew product systems is well-established in cases where the fiber is one-dimensional. This includes a classification of the number of invariant graphs for specific classes of skew product systems \cite{jager2003quasiperiodically, fadaei2018invariant}.
\begin{definition}[Invariant graph]\label{def:graph}

Let $F_{a,b} \in \mathcal{F}$. A measurable function $\phi : \mathbb{I} \to \mathbb{I}$ is called an invariant graph (with respect to $F_{a,b}$) if for all $x \in \mathbb{I}$:
$$ F_{a,b}(x, y) = (f(x),  \phi(f(x) )) ,    \text{ or equivalently} \quad g_{a,b}(\phi(x)) = \phi(f(x))  .
$$
The point set $\Phi \coloneqq \{ (x, \phi(x))  : x \in \mathbb{I} \}$ will also be referred to as the invariant graph, labeled with the corresponding capital letter.  
Denote by $\mathcal{M}_f$ the space of all $f$-invariant probability measures and by $\mathcal{E}_f$ the family of all ergodic $f$-invariant measures.
If $\nu \in \mathcal{M}_f$ and the identity holds for $\nu$-almost every $x \in \mathbb{I}$ we call $\phi$ a $\nu$-a.\,e.\@ invariant graph.
\end{definition}

\begin{remark}\label{bo}
It is easy to see that each $F_{a,b} \in \mathcal{F}$ has two constant invariant bounding graphs $\phi^0(x) \coloneqq 0$ and $\phi^1(x) \coloneqq1$. 
Therefore, the corresponding point sets $\Phi_0=\{ (x,\phi^0(x)): x \in \mathbb{I} \}=\mathbb{I} \times \{0\}$ and  $\Phi_1=\{ (x,\phi^1(x)): x \in \mathbb{I} \}=\mathbb{I} \times \{1\}$ are invariant sets. They play the role of the $F_{a,b}$-invariant manifolds with chaotic dynamics inside.
\end{remark}

\begin{remark}
Each skew product system $F_{a,b}$ has monotone fiber map $g_{a,b}$ which possess negative Schwarzian derivatives meaning that
\[
\dfrac{g_{a,b}'''}{g_{a,b}'} - \frac{3}{2} \left( \frac{g_{a,b}''}{g_{a,b}'} \right) ^{2} < 0,
\]
for all $a,b \in (0,1/2)$. Note that in our setting this condition ensures that there are at most three invariant graphs. For this reason we will focus on the bounding graphs $\phi^0$ and $\phi^1$.
\end{remark}

In what follows, let $\Phi_i$, $i=0,1$, denote the $1$-dimensional subspaces $\mathbb{I} \times \{i\}$, as defined in Remark \ref{bo}, which are forward invariant under $F_{a,b}$. 
In our model, the base map $f$ is an expanding Markov map that admits a chaotic attractor.  This attractor supports an absolutely continuous invariant ergodic measure  $\nu_{\mathrm{ac}}$  whose density is bounded and bounded away from zero (see \cite{viana2016foundations}).
In particular, this measure is absolutely continuous with respect to Lebesgue.
Due to this fact and the invariance of the subspace $\Phi^i$, the restriction of $F_{a,b}$ to the corresponding invariant subspaces possesses a chaotic attractor 
\begin{equation*}
    \label{in1}
A_i\;\mbox{ with  the basin of attraction }\; \mathcal{B}(A_i), \ i=0,1.
\end{equation*}

In particular, these attractors are SRB attractors \cite{ashwin1996attractor}.
 We consider the restriction of $F_{a,b}$ to $\Phi_i$, denoted by $F_{a,b}|_{\Phi_i}$.
Let us denote by $\mathcal{M}_{F_{a,b}}(A_i)$  the sets of $F_{a,b}$-invariant probability measures  supported on $A_i$ and by  $\mathcal{E}_{F_{a,b}}(A_i)$ its subset of  ergodic measures. 

Here, we have two kinds of Lyapunov exponents for the  invariant set $A_i$: the \emph{parallel Lyapunov exponents} and the \emph{normal Lyapunov exponents}.
For a map $F$ defined on any of the relevant sets we let $dF$ denote its total derivative.

\begin{definition} 
Given $(x, \phi^i(x)) \in A_i$, $i=0,1$, we define the \emph{parallel Lyapunov exponent} at point $(x,\phi^i(x))$ to be
\begin{eqnarray}\label{22}
\lambda_{\parallel}(x,\phi^i(x)) =\lim_{n\rightarrow \infty} \dfrac{1}{n} \log \left| d_xF^n_{a,b}|_{ \Phi_i}(x,\phi^i(x))\right|.
\end{eqnarray}
In our setting, $\lambda_{\parallel}(x,i)=\lim_{n\rightarrow \infty} \dfrac{1}{n} \log |df(x)|$, for $i=0,1$, which is positive.
Similarly, we define the \emph{normal Lyapunov exponent} at $(x,\phi^i(x))$ to be
\begin{equation}\label{11}
  \lambda_{\perp, a,b}(x,\phi^i(x))=\lim_{n \to \infty}\frac{1}{n}\log \left|d g_{a,b,x}^n(\phi^i(x))\right|,
\end{equation}
whenever the limit exists.
\end{definition}

Note that the local dynamic stability of chaotic attractors within invariant submanifolds can be characterized using their normal Lyapunov exponents. These exponents represent the additional stability measures introduced when the attractor is considered as a subset of the global phase space, rather than being confined solely to the invariant submanifold \cite{ashwin1996attractor}. 

\begin{remark}\label{ex}
    Let $\mu$ be an $F_{a,b}$-invariant ergodic measure supported in $A_i$. Then, for $\mu$-a.\,e.\@ $(x,\phi^i(x)) \in \mathbb{I} \times \mathbb{I},$  the normal Lyapunov exponent  $\lambda_{\perp, a,b}(x,\phi^i(x))$ exists (see \cite[Theorem~2.3]{ashwin1996attractor}).
   Furthermore, since $\mu$ is ergodic,  for $\mu$-a.\,e.\@ $(x,\phi^i(x)) \in \mathbb{I} \times \mathbb{I},$  the normal Lyapunov exponent  $\lambda_{\perp, a,b}(x, \phi^i(x))$ is constant denoted by $\lambda_{\perp, a,b}(\mu)$.
 \end{remark}
\begin{definition}
Let $\nu$ be an $f$-invariant ergodic measure. If $\phi$ is a $\nu$-a.\,e.\@ invariant graph with
$\log |dg_{a,b,x}(\phi(x))| \in  \mathcal{L}^1_\nu$, then its Lyapunov exponent w.\,r.\,t.\@ $\nu$ is defined as
\begin{equation}\label{111}
  \lambda_{\nu, a,b}(\phi)\coloneqq \int_{\mathbb{I}}\log \left|d g_{a,b,x}(\phi(x))\right|\d\nu(x).
\end{equation}
\end{definition}
\begin{remark}\label{BI}
Let $\phi$ be a $\nu$-a.\,e.\@ invariant graph. Note that by the Birkhoff ergodic theorem the following holds:
\begin{equation*}
\begin{aligned}
 \lambda_{\perp, a,b}(x,\phi(x)) 
&= \lim_{n\to \infty} \frac{1}{n} \log |dg_{a,b,x}^{n}(\phi(x))| = \lim_{n\to \infty} \frac{1}{n} \sum_{k=0}^{n-1}\log |dg_{a,b,f^{k}(x)}(g_{a,b,x}^{k}(\phi(x)))|\\
&= \lim_{n \to \infty}\frac{1}{n} \sum_{k=0}^{n-1}\log |dg_{a,b,f^{k}(x)}(\phi(f^k(x)))|
= \int_{\mathbb{I}} \log |dg_{a,b,x}(\phi(x))|\d\nu (x)  \\
&=\lambda_{\nu, a,b}(\phi)
\end{aligned}
\end{equation*}
for $\nu$-a.e.\@ $x\in\mathbb{I}$.
\end{remark}

Here, we focus on a specific class of invariant measures known as Sinai-Ruelle-Bowen SRB measures  \cite{pugh1989ergodic}. 
We define $\mu$ as an SRB measure for $A_i$, $i=0, 1,$ if it is an invariant ergodic probability measure supported on $A_i$ and possesses absolutely continuous conditional measures on unstable manifolds (with respect to the Riemannian measure). In our model, the base map $f$ is an expanding map, meaning the chaotic attractors $A_i$ support an absolutely continuous invariant ergodic measure equivalent to the Lebesgue measure. 

The attractor $A_i$ is an SRB attractor if it supports an SRB measure.
Since $A_i$ is an asymptotically stable attractor under $F_{a,b}|_{\Phi_i}$, for $i=0$ or $1$, it is the closure of the union of unstable manifolds. Note that the existence of an SRB measure supported on $A_i$ implies the absolute continuity of the stable foliation of $A_i$, see \cite{pugh1989ergodic}.

By Remark \ref{ex}, for a given ergodic invariant probability measure $\mu \in \mathcal{E}_{F_{a,b}}(A_i)$, the normal Lyapunov exponent $\lambda_{\bot, a,b}(\mu)$ exists and is constant in a set 
of full $\mu$-measure. For simplicity, we set $\lambda_{\bot}(\mu)\coloneqq \lambda_{\bot, a,b}(\mu)$.
We define
\begin{equation}\label{mm}
  \lambda_{\min}(A_i)\coloneqq\inf \left\{\lambda_{\bot}(\mu):\mu \in \mathcal{E}_{F_{a,b}}(A_i) \right\}. 
  \end{equation}
Let $\mu$ be an $F_{a,b}$-invariant ergodic probability measure supported in $A_i$, with normal Lyapunov exponents $\lambda_{\bot}(\mu)$. In our setting, since there is only one normal direction, based on \cite{ashwin1996attractor}, we set
\begin{equation}\label{in}
 \Lambda_\mu\coloneqq\lambda_{\bot}(\mu).
\end{equation}
 We recall the next result from \cite{ashwin1996attractor}.
\begin{proposition}\label{thm11}
Assume $F$ is a skew product with invariant subspace $\Phi$ and $A$ is an SRB attractor for $F|_\Phi$ with $\Lambda_{\mathrm{SRB} }< 0$,
where $\Lambda_{\mathrm{SRB}}$ is defined by (\ref{in}) for the SRB measure $\mu_{\mathrm{SRB}}$.
Then $m(\mathcal{B}(A)) > 0$. Furthermore, $A$ is an essential attractor provided that $A$ is either uniformly hyperbolic or $\mu_{\mathrm{SRB}}$ is absolutely continuous with respect to the Riemannian measure on $\Phi$.
\end{proposition}

Note that the parameters $a$ and $b$ vary the transverse dynamics without changing the dynamics on the invariant subspace $\Phi_i$.

Let
\begin{equation}\label{region1}
 \Gamma_0\coloneqq\left\{(a,b) : 0<a<1/2, \ 0< b <1/2, \ b > a/(1+a)  \right\},
\end{equation}
\begin{equation}\label{region2}
 \Gamma_1\coloneqq\left\{(a,b) : 0<a<1/2, \ 0< b <1/2, \ b<a/(1-a)  \right\},
\end{equation}
and
\begin{equation}\label{region3}
 \Gamma\coloneqq\Gamma_0 \cap \Gamma_1= \left\{(a,b) : 0<a<1/2, \ 0< b <1/2, \ b > a/(1+a), \ b<a/(1-a)  \right\}.
\end{equation}
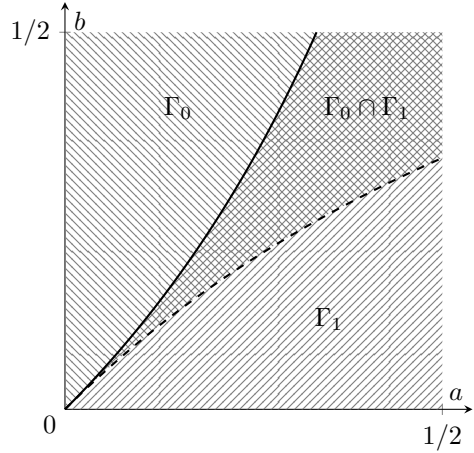
\begin{figure}[!ht]
\begin{center}
\begin{tikzpicture}
    \begin{axis}[
        axis lines=middle,
        width=7cm,
        height=7cm,
        xmin=0, xmax=0.54,
        ymin=0, ymax=0.54,
        xlabel={$a$},
        ylabel={$b$},
        samples=100,
        domain=0:0.5,
        xtick={0, 0.5},
        ytick={0, 0.5},
        xticklabels={$0$, ${1}/{2}$},
        yticklabels={$0$, ${1}/{2}$},
    ]

    \addplot [thick, dashed, name path=lower]
    {x / (1 + x)};

    \addplot [thick, name path=upper,domain=0:1/2,color=black,opacity=0]
    {min(0.5,x/(1-x))};
\addplot [thick,domain=0:1/3]
    {min(0.5,x/(1-x))};

    \path[name path=axis1] (0,0) -- (0.5,0);
    \path[name path=axis2] (0,0.5) -- (0.5,0.5);

    % Fill the region above the lower curve (Gamma_1)
    \addplot [
        pattern=north west lines,
        pattern color=black!50
    ] fill between[
        of=lower and axis2,
        soft clip={domain=0:0.5}
    ];
    \node at (axis cs:0.35,0.12) {$\Gamma_1$};

    % Fill the region below the upper curve (Gamma_0)
    \addplot [
        pattern=north east lines,
        pattern color=black!50
    ] fill between[
        of=axis1 and upper,
        soft clip={domain=0:0.5}
    ];
    \node at (axis cs:0.15,0.4) {$\Gamma_0$};
    \node at (axis cs:0.4,0.4) {$\Gamma_0\cap\Gamma_1$};

    \end{axis}
    \node at (-0.2,-0.2) {$0$};
\end{tikzpicture}
\end{center}
\caption{The dashed curve represents the function  $a\mapsto{a}/({1+a})$ and the solid black curve represents the function $a\mapsto{a}/{(1-a)}$. The area above the lower and below the upper curve are hatched differently and define the regions  $\Gamma_0$ and $\Gamma_1$ in the $ab$-plane. The intersection of these areas are denoted by $\Gamma$.
  For each $(a,b)\in  \Gamma$,  both invariant sets $A_0$ and $A_1$ of the skew product system $F_{a,b}$ are Milnor attractors with locally riddled basins.}
   \label{fig:2}
   \end{figure}
   Now we are in the position to state our first result; its  proof will be postponed to in Section \ref{section3}.
\begin{theorem}\label{T1} 
For $F_{a,b} \in \mathcal{F}$, the following  statements hold. 
\begin{enumerate}
    \item If  $(a, b) \in \Gamma_0$ then $F_{a,b}$ admits a chaotic  Milnor (essential) attractor $A_0 \subset \Phi_0$ with a locally riddled basin $\mathcal{B}(A_0)$.
    \item If  $(a, b) \in \Gamma_1$ then $F_{a,b}$ admits a chaotic   Milnor (essential) attractor $A_1 \subset \Phi_1$ with a locally riddled basin $\mathcal{B}(A_1)$.
    \item If $(a, b) \in \Gamma$ then $F_{a,b}$ admits two chaotic (essential) Milnor attractors $A_0$ and $A_1$ so that their basins are intermingled. Moreover, there is an invariant measurable graph $\phi^*: \mathbb{I} \to \mathbb{I}$ that separates the basins $\mathcal{B}(A_i)$.
\end{enumerate}
 
 \end{theorem}
\begin{remark}\label{R1}
         In the proof of Theorem \ref{T1} we demonstrate that two attractors $A_i$, $i=0,1,$ have negative normal Lyapunov exponent (see Section \ref{section3}).  Additionally, $A_i$ are Milnor essential attractors 
within the invariant sets $\Phi_i$. By the definition of an essential attractor and Remark \ref{BI}, two bounding graphs $\phi^i$, $i=0,1,$ are $m$-a.\,e.\@ invariant graphs.
Since, the $f$-invariant measure $\nu_{\mathrm{ac}}$ is equivalent to the Lebesgue measure, the two bounding graphs $\phi^i$, $i=0,1,$ are also $\nu_{\mathrm{ac}}$-almost everywhere invariant graphs.

\end{remark}

\begin{definition}
    The basins $\mathbb{B}_i$, $i=0,1,$ of $\nu_{\mathrm{ac}}$-almost everywhere invariant graphs $\phi^i$ are defined as: 
    \begin{equation}
        \mathbb{B}_i\coloneqq\{ (x,y ) \in \mathbb{I}\times \mathbb{I}: g_{a,b,x}^n(y) \to i, \ as \ n \to \infty \}.
    \end{equation}
        \end{definition}
By Remark \ref{R1} and Theorem \ref{T1},  the basins $\mathbb{B}_i$ of $\nu_{\mathrm{ac}}$-almost everywhere invariant graphs $\phi^i$ and $\mathcal{B}(A_i)$ of Milnor attractors $A_i$, $i=0,1,$ coincide $\nu_{\mathrm{ac}}$-almost everywhere.   
 So, by Theorem \ref{T1}, Remark \ref{R1} and \cite[Proposition~1.6]{keller2017stability}, we have the following result.

 \begin{corollary}\label{C1}
     For $(a,b) \in \Gamma$ the map  $F_{a,b}\in \mathcal{F}$ has three invariant graphs $\phi^0 < \phi^* < \phi^1,$ such that $(\nu_{\mathrm{ac}} \times m)(\mathbb{I} \times \mathbb{I} \setminus (\mathbb{B}_0 \cup \mathbb{B}_1))=0$ and $(\nu_{\mathrm{ac}} \times m)(\mathbb{B}_i) >0,$ $i=0,1$. 
 \end{corollary}

When the chaotic attractor loses its stability, it gives rise to a chaotic saddle.
 A chaotic saddle is a set of states that the system can transiently visit, but it is not an attractor (i.e., the system doesn’t settle into these states permanently). It can act as a kind of intermediary state, influencing the system's long-term dynamics. 
\begin{theorem}\label{T2}
 Let $F_{a,b} \in \mathcal{F}$. Then the following holds:
 \begin{enumerate}
\item  If  $0<b<{a}/({1+a})$ and $0<a<1/2$, then $F_{a,b}$ admits two invariant sets $A_i$, $i=0,1,$  so that $A_0$ is a chaotic saddle and $A_1$ is a chaotic Milnor attractor with a locally riddled basin.
  \item  If ${a}/{(1-a)}<b<1/2$ and $0<a<1/2$, then $F_{a,b}$ admits two invariant sets $A_i$, $i=0,1,$  so that $A_1$ is a chaotic saddle and $A_0$ is a chaotic Milnor attractor with a locally riddled basin.
\end{enumerate}
\end{theorem}

If $A$ is a chaotic attractor that supports an SRB measure $\mu_{\mathrm{SRB}}$ then the sign of $\Lambda_{\mathrm{SRB}}$ determines the transverse behavior of infinitesimal perturbations relative to the invariant set  $\Phi$.
When $\Lambda_{\mathrm{SRB}}<0$, $A$ attracts trajectories transversely in phase space, making it an attractor for the entire phase space.   Conversely, if $\Lambda_{\mathrm{SRB}}>0$,
trajectories near $A$ are repelled transversely, rendering $A$ transversely unstable and not an attractor for the entire phase space. 

Based on the above observation, a bifurcation occurs when $\Lambda_{\mathrm{SRB}}$ crosses zero, known as a \emph{blowout bifurcation}.  

Blowout bifurcations can be categorized into two types based on the system's behavior near the bifurcation point: \emph{subcritical} (\emph{hysteretic})  and \emph{supercritical} (\emph{non-hysteretic}). 
A subcritical blowout bifurcation occurs when the invariant subspace becomes transversely unstable ($\Lambda_{\mathrm{SRB}}>0$), and no nearby attractors emerge to replace it.  In contrast, a supercritical blowout results in a soft loss of stability, transitioning to an on-off intermittent attractor \cite{ashwin1998unfolding, platt1993off}. 

By Theorem \ref{T1}, Theorem \ref{T2}, and by the definition, the following result is evident.

 \begin{corollary}\label{C2}
   Let $F_{a,b}\in \mathcal{F}$ be a skew product of the form (\ref{ss}) whose fiber maps $g_{a,b}$  given by (\ref{fiber1}).
 Then the following holds:
\begin{itemize}
\item[$(a)$] $F_{a,b}$ exhibits a (subcritical) hysteretic blowout bifurcation on passing
through any $0<a <1/2$ and $ b ={a}/(1-a)$;
\item[$(b)$]  $F_{a,b}$ exhibits a (supercritical) non-hysteretic blowout bifurcation on passing through any $0< a <1/2$ and $b =  {a}/{(1+a)}$.
\end{itemize}
 Specifically, there exist three distinct dynamical regimes $\Gamma$, $\Gamma_0 \setminus \Gamma$, and $\Gamma_1 \setminus \Gamma$, which are typically separated by bifurcation curves where stability changes.
 \end{corollary}
 
Note that dynamic regimes of the two-parameter map $F_{a,b}$ can be understood as the qualitative behaviors of the system's trajectories in its state space as the parameters $a$ and $b$ vary.  In particular,  the system shows a sensitive dependence on initial conditions, and trajectories appear aperiodic and unpredictable. 

The following condition will be used frequently from now on. 

\begin{condition*} [\textbf{H1}] For $(a,b) \in \Gamma$, there exists an equilibrium state $\nu_\psi$ corresponding to the Hölder continuous potential $\psi$ with $P(\psi) = 0$ and  such that $\lambda_{\nu_\psi,a,b}(\phi^i)<0$, $i=0,1$.  
\end{condition*}

 Note that if we take $\psi=- \log |df|$, then $\nu_\psi$ coincides with the SRB measure $\nu_{\mathrm{ac}}$. Clearly, $\psi$ is normalized. Moreover, by the proof of Theorem \ref{T1}, $\lambda_{\nu_{ac,a,b}}(\phi^i)<0$.

We mention that in our setting the function $(x,y) \mapsto \log dg_{a,b,x} (y) $ restricted to each of the sets $I_i \times \mathbb{I}$, $i=0,1,$ is Hölder continuous, as mentioned in \cite[Hypothesis~3]{keller2017stability}, where $I_i$ is a Markov interval of $f$. In fact, $g_{a,b,x}=g_a$ on $I_0 \times \mathbb{I}$ and $g_{a,b,x}=g_b$ on $I_1 \times \mathbb{I}$.

Note that for our family there exist two $f$-invariant probability measures $\nu^i$, $i=0,1,$ such that $\lambda_{\nu^i,a,b}(\phi^i)>0$ (see the proof of Theorem \ref{T1}).

\begin{remark}
  Since $(a,b) \in \Gamma$, the basins of $A_i$ are intermingled with respect to $\nu_{\mathrm{ac}}$. Furthermore, $\nu_\psi$ is an ergodic Gibbs measure, by condition (H1)  and  \cite[Proposition~1.6]{keller2017stability}, the basins of $A_i$ are intermingled with respect to $\nu_\psi$ and the following holds:
$(\nu_\psi \times m )((\mathbb{I}\times \mathbb{I})\setminus (\mathbb{B}_0 \cup \mathbb{B}_1))=0$ and $(\nu_\psi \times m) (\mathbb{B}_i)>0$, for $i=0,1.$ 
  
\end{remark}

We generalize \cite[Theorem~1]{keller2017stability} in the following way.
\begin{theorem}\label{T3}
    
  Let us assume that for $(a,b)\in \Gamma$ condition (H1)  is satisfied and that the Gibbs measure $\nu_\psi$ with respect to $F_{a,b}$ is given as stated therein. Then there are \( t_{0}^* \) and \( t_{1}^* > 0 \) such that 
\begin{equation*}
\lim_{\epsilon \to 0} \frac{\log  \nu_\psi\{\phi^* < \phi^0 +\epsilon\}}{\log \epsilon} = t_{0}^*,  \quad     \lim_{\epsilon \to 0} \frac{\log  \nu_\psi\{\phi^* >\phi^1 - \epsilon\}}{\log \epsilon} = t_{1}^*.
\end{equation*}
For $i=0,1$, the number \( t_i^* \), is uniquely determined as the  positive zero of the pressure function
\[
t \mapsto p_{i,\psi}(t) \coloneqq P(\psi + t \log dg_{a,b}(\phi^i)).
\]
\end{theorem}

\begin{definition}
Let $\nu$ denote an appropriate $f$-invariant probability measure on $\mathbb{I}$ and $m$ the Lebesgue measure on $\mathbb{I}$.   
For $i=0,1$, the  \emph{local stability index} $\sigma_\nu(x,y)$ of a point $(x,y) \in \mathbb{I} \times \mathbb{I}$ with respect to $\mathbb{B}_i$ is defined as follows: 
\begin{equation}
\Sigma_{\nu, \epsilon}^i(x,y)\coloneqq \frac{\nu\times m (U_\epsilon(x,y) \cap \mathbb{B}_i)}{\nu\times m(U_\epsilon(x,y))},
\end{equation}
where $U_\epsilon(x,y)=[x-\epsilon,x+\epsilon] \times [y-\epsilon,y+\epsilon]$.
Further, with 
\begin{equation}
    \sigma_\nu^i(x,y)\coloneqq\lim_{\epsilon \to 0}\frac{\log \Sigma^i_{\nu,\epsilon}(x,y)}{\log \epsilon},
\end{equation}
we set
\begin{equation}
    \sigma_\nu(x,y)\coloneqq \sigma_\nu^1(x,y) - \sigma_\nu^0(x,y).
\end{equation}
\end{definition}
Clearly, $\Sigma_{\nu,\epsilon}^0(x,y) +\Sigma_{\nu,\epsilon}^1(x,y) =1$.  
Under condition (H1), we determine the stability index for $\nu_\psi$-a.\,e.\@ point and the multifractal spectrum of the stability index also with respect to $\nu_\psi$.

\begin{theorem}\label{T4}
    Let $F_{a,b} \in \mathcal{F}$ with $(a,b) \in \Gamma$, $F_{a,b}$ fulfills the condition (H1) with $\nu_\psi$ specified therein. Then the following statements hold:
    \begin{enumerate}
        \item The Basin $\mathbb{B}_1$ is riddled with $\mathbb{B}_0$, and for $\nu_{\psi}$-a.\,e.\@ $x \in \mathbb{I}$ and all $y>\phi^*(x),$ one has that
\begin{equation*}
  \sigma_{\nu_\psi}(x,y)=  -\sigma_{\nu_\psi}^0(x,y)=t_1^*\cdot \frac{ \lambda_{\nu_{\psi},a,b}(\phi^1)}{\int \log |df| \d\nu_{\psi}}<0,
\end{equation*}
where $t_1^*$ is as in Theorem \ref{T3}.
        \item The Basin $\mathbb{B}_0$ is riddled with $\mathbb{B}_1$ and for $\nu_{\psi}$-a.\,e.\@ $x \in \mathbb{I}$ and all $y<\phi^*(x)$, one has that
\begin{equation*}
  \sigma_{\nu_\psi}(x,y)=  \sigma_{\nu_\psi}^1(x,y)=t_0^*\cdot  \frac{ -\lambda_{\nu_{\psi},a,b}(\phi^0)}{\int \log |df| \d\nu_{\psi}}>0,
\end{equation*}
where $t_0^*$ is as in Theorem \ref{T3}.
    \end{enumerate}
\end{theorem}

Note that, since $(a,b) \in \Gamma$, by Theorem \ref{T1}, the basins $\mathbb{B}_0$ and $\mathbb{B}_1$ are intermingled.
\begin{remark}
    When $\psi=-\log |df|$, we have $P(\psi)=0$. In this case, the 
Loynes exponent $t_i^*$ is defined by $p_{i, \psi}=P(- \log |df| + t_i^*  \log dg_{a,b}(\phi^i) )= 0$. Consequently, Theorem \ref{T3} and Theorem \ref{T4} yield   \cite[Theorems 1 and 2]{keller2017stability}.

\end{remark}
Multifractal analysis provides a framework within the thermodynamic formalism to study the fine-scale geometric and dynamical properties of measures, particularly their local scaling variability under potential functions. \cite{pesin1997multifractal, kessebohmer2008fractal, kessebohmer2007multifractal, kessebohmer2004multifractal}.  
Multifractal analysis can be used to investigate the complex geometric and dynamic structures of basins of attraction in nonlinear systems, especially riddled and intermingled basins. These basins arise in systems that are characterized by a sensitive dependence on the initial conditions and intricate boundary structures. 

In the following,  let $F_{a,b} \in \mathcal{F}$ with  $(a,b) \in \Gamma$ and $\nu=\nu_\psi$ as defined in condition (H1). Then the basins of chaotic attractors $A_i$ are intermingled by $\nu$. We calculate the multifractal spectrum of the stability index with respect to $\nu$. We define
\begin{equation}\label{M0}
 A_{0,\nu}(\sigma)=\{x \in \mathbb{I}: \sigma_\nu(x,y)=\sigma, \ \text{for} \ \text{every} \ y<\phi^*(x)\},   
\end{equation}

and
\begin{equation}\label{M1}
    A_{1,\nu}(\sigma)=\{x \in \mathbb{I}: \sigma_{\nu}(x,y)=-\sigma, \ \text{for} \ \text{every} \ y>\phi^*(x)\}.
\end{equation}

\begin{theorem}\label{T5}
    Take $T(q)$ as  $P(-T(q) \log  |df| + q t_0^*\log dg_{a,b}(\phi^0) ) = 0$, and let $\nu_q$ be the equilibrium state with potential $-T(q) \log  |df| + q t_0^* \log dg_{a,b}(\phi^0)$, where $t_0^*$ is given by Theorem \ref{T3}, and let
    \begin{equation*}
        \sigma(q)=-dT(q)=-t_0^* \int \log dg_{a,b}(\Phi^0)\d\nu_q / \int \log|df|\d\nu_q.
    \end{equation*}
  Then, $\nu_0=\nu_{\mathrm{ac}}$, $T(0)=1$ and $T(q)$ is strictly convex function. 
Moreover, the following hold:
\begin{enumerate}
\item There exists a unique $q^* \in (0, 1)$ such that $\int \log dg_{a,b}(\phi^0)\d\nu_{q^*}=0$. 
\item The functions $\sigma \mapsto A_{0,\nu}(\sigma)$ and $q \mapsto T(q)$ form a Legendre transform pair. In particular
\begin{equation*}
    \dim_H (A_{0,\nu}(\sigma(q)))=T(q) + q \sigma(q), \ \text{for \ all} \ q>q^*.
\end{equation*}
 \end{enumerate}   
 \end{theorem}

\begin{theorem}\label{T6}
   Take $S(q)$ as  $P(-S(q) \log  d|f| + q t_1^*\log dg_{a,b}(\phi^1) ) = 0$ and let $\nu_q$ be the equilibrium state with potential $-S(q) \log  |df| + q t_1^* \log dg_{a,b}(\phi^1)$,  where $t_1^*$ is given by Theorem \ref{T3}, and  let
    \begin{equation*}
        \sigma(q)=-dS(q)=-t_1^* \int \log dg_{a,b}(\Phi^1)\d\nu_q \bigg/ \int \log|df|\d\nu_q.
    \end{equation*}
   Then, $\nu_0=\nu_{\mathrm{ac}}$, $S(0)=1$, and $S(q)$ is strictly convex function. 
Moreover, the following hold:
\begin{enumerate}
\item There exists a unique $q^* \in (0, 1)$ such that $\int \log dg_{a,b}(\phi^1)\d\nu_{q^*}=0$. 
\item The functions $\sigma \mapsto A_{1,\nu}(\sigma)$ and $q \mapsto S(q)$ form a Legendre transform pair. In particular
\begin{equation*}
    \dim_H (A_{1,\nu}(\sigma(q)))=S(q) + q \sigma(q), \, \text{ for all }  q<q^*.
\end{equation*}
 \end{enumerate}   
 \end{theorem}

\section{Locally riddled basin and chaotic saddle} \label{section3}
Let $A$ be a chaotic Milnor attractor of $C^{1+\alpha}$ map $F$ defined on a smooth manifold $M$. Given an ergodic measure $\mu \in \mathcal{E}_{F}(A)$, let $G_\mu$ be the set of generic points of $\mu$. That is
 \begin{equation*}
   G_\mu=G_\mu(A)=\left\{(x,y) \in A:\frac{1}{n}\sum_{i=0}^{n-1}\delta_{F_{a,b}^j(x,y)}\to \mu\right\}
 \end{equation*}
 where convergence is in the weak$^*$ topology. For any $\alpha > 0$ define
 \begin{equation}\label{alpha}
   G_\alpha=G_{\alpha}(A)\coloneqq\bigcup_{\mu \in \mathcal{E}_{F}(A),\Lambda_\mu \geq \alpha}G_\mu.
 \end{equation}
We recall the following  result \cite[Proposition~3.19]{ashwin1996attractor}.
\begin{proposition}\label{p1}
Suppose $F : M \to M$ is a $C^{1+\alpha}$ map leaving the embedded submanifold $\Phi$ invariant, and that $A$ is an asymptotically stable chaotic attractor for $F|_{\Phi}$.
Let $\Lambda_{\mathrm{max}}$, $\lambda_{\mathrm{min}}$ and $\Lambda_{\mathrm{SRB}}$ be given by (\ref{mm}) and (\ref{in}). Then, under $F : M \to M$
\begin{enumerate}
  \item [$(1)$] If $\Lambda_{\mathrm{SRB}} < 0 < \Lambda_{\mathrm{max}}$ and there
exists $\alpha > 0$ with $G_\alpha$ dense in $A$, then $A$ is a Milnor attractor with a locally riddled basin.
  \item [$(2)$] If $ \lambda_{\mathrm{min}}< 0 < \Lambda_{\mathrm{SRB}}$, $\mu_{\mathrm{SRB}}$-almost all Lyapunov exponents are non-zero and $m(\bigcup_{\mu \neq \mu_{\mathrm{SRB}}}G_\mu)=0$, where $m$ is the Riemannian volume on $\Phi$,  then $A$ is a chaotic saddle.
\end{enumerate}
\end{proposition}

It is important to note that the normal dynamics vary continuously with the parameters $a$ and $b$. Furthermore, the invariant subspaces $\Phi^i$, $i=0,1$, have codimension 1 within the phase space $\mathbb{I} \times \mathbb{I}$, possessing only a single normal direction. 
By \cite[Remark~3.4]{ashwin1994bubbling}, if $\mathrm{codim}(\Phi)=1$, as in our setting, there is only one normal direction. In this case, $\lambda_\mu=\Lambda_\mu$ for all ergodic $\mu$, and the normal spectrum depends smoothly on normal parameters.

\begin{remark}\label{rem-dis}
In our model, the base map $f$ has a discontinuity at $x=1/2$. Consequently, the set of discontinuity points of  $F_{a,b}$ has measure zero. It is straightforward to verify that the conclusion of Proposition \ref{p1} holds when $F$ has a zero-measure set of discontinuity points \cite{ashwin1996attractor}.

\end{remark}
 \subsection{Proof of Theorem \ref{T1}}
Using Proposition \ref{p1}, we begin to prove Theorem \ref{T1}.

Let $F_{a,b} \in \mathcal{F}$. Then the following  statements hold: 
\begin{enumerate}
    \item If  $(a, b) \in \Gamma_0$ then $F_{a,b}$ admits a chaotic (essential) Milnor attractor $A_0 \subset \Phi_0$.
    \item If  $(a, b) \in \Gamma_1$ then $F_{a,b}$ admits a chaotic (essential)  Milnor attractor $A_1 \subset \Phi_1$.
\end{enumerate}
Indeed, let $\Phi_0$  be given by (\ref{subspaces}), $(a,b)\in \Gamma_0$ and consider the restriction $F_{a,b}|_{\Phi_0}$.
Since there is only one normal direction,  $\Lambda_{a,b, \mu}(A_0)$ of an ergodic invariant probability measure  $\mu \in \mathcal{E}_{F_{a,b}}(A_0)$ is equal to the normal Lyapunov exponent. For simplicity, we write $\Lambda_\mu(A_0)\coloneqq\Lambda_{a,b, \mu}(A_0)$.
Note that, for $0 \leq x < 1/2$, we have
\begin{equation*}
 d_{(x,0)}F_{a,b}=\begin{pmatrix}
2 & 0\\
0& dg_{a}(0)\\
\end{pmatrix}
=\begin{pmatrix}
2 &  0\\
0& 1+a\\
\end{pmatrix},
\end{equation*}
and for $1/2 \leq x \leq 1$, we have
\begin{equation*}
 d_{(x,0)}F_{a,b}=\begin{pmatrix}
2 &  0\\
0&  dg_{b}(0)\\
\end{pmatrix}
=\begin{pmatrix}
2 &  0\\
0&  1-b\\
\end{pmatrix},
\end{equation*}
where $g_{a}$ and $g_b$, given by (\ref{fiber1}).

Hence, 
\begin{eqnarray}\label{Lambd}
  \Lambda_{\mu}(A_0) = \int_{A_0 \cap ([0,1/2]\times \mathbb{I})}\log (1+a) \d\mu(x,y ) + \int_{A_0 \cap( (1/2,1]\times \mathbb{I})}\log (1-b) \d\mu(x,y ).
   \end{eqnarray}

Note that for each invariant measure $\mu$, $\Lambda_{\mu}$ is finite.
Additionally, $\Lambda_{\mu}$ is smoothly dependent on the parameters $a$ and $b$.
The base map $f$  is a piecewise expanding map.
By definition of $F_{a,b}$ and since $\Phi_0$ is one dimensional, we conclude that
$F_{a,b}|_{A_0}$ is also piecewise expanding. This fact implies that $F_{a,b}|_{A_0}$ has a
an ergodic invariant measure equivalent to Lebesgue (see \cite{ashwin1994bubbling, walters2000introduction}); this corresponds to the desired $\mu_{\mathrm{SRB}}(A_0)$ (see Subsection 4.3 of \cite{ashwin1994bubbling}).
By this fact 
\begin{equation*}
\Lambda_{\mathrm{SRB}}(A_0)=1/2 \log (1+a) + 1/2 \log (1-b).
\end{equation*}
Note that $\Lambda_{\mathrm{SRB}}(A_0)$ is the normal Lyapunov exponent on $A_0$ denoted by $L_{\perp, a,b}(0)$ .
It characterizes \cite{ashwin1996attractor} evolution transverse to the $x$-axis. If it is negative, the invariant set $A_0$ is a Milnor essential attractor.
Simple computations show that, for $0<a<1/2$, $0<b<1/2$, and $b> a/(1+a)$ $,$  we have $\Lambda_{\mathrm{SRB}}(A_0)<0$,  and hence, by Proposition \ref{thm11}, $A_0$ is a Milnor (essential) attractor.

The same argument is applied for $A_1$.
Indeed, for $0 \leq x \leq 1/2$, we have
\begin{equation*}
 d_{(x,1)}F_{a,b}=\begin{pmatrix}
2 &  0\\
0&  dg_{a}(1)\\
\end{pmatrix}
=\begin{pmatrix}
2 &  0\\
0&  1-a\\
\end{pmatrix},
\end{equation*}
and for $1/2 < x \leq 1$, we have
\begin{equation*}
 d_{(x,1)}F_{a,b}=\begin{pmatrix}
2 &  0\\
0&  dg_{b}(1)\\
\end{pmatrix}
=\begin{pmatrix}
2 &  0\\
0&  1+b\\
\end{pmatrix},
\end{equation*}
where $g_{a}$, $g_b$, given by (\ref{fiber1}).
Hence, 
\begin{eqnarray}\label{1Lambd}
  \Lambda_{\mu}(A_1) = \int_{A_1 \cap ([0,1/2]\times \mathbb{I})}\log (1-a)\d\mu(x,y ) + \int_{A_1 \cap ((1/2,1]\times \mathbb{I})}\log (1+b)\d\mu(x,y ).
   \end{eqnarray}
As stated above, $F_{a,b}|_{A_1}$ possesses an ergodic invariant measure that is equivalent to the Lebesgue measure; this corresponds to the desired $\mu_{\mathrm{SRB}}(A_1)$.
By this fact,  
\begin{equation}\label{L2}
\Lambda_{\mathrm{SRB}}(A_1)=1/2 \log (1-a) + 1/2 \log (1+b).
\end{equation}
Note that $\Lambda_{\mathrm{SRB}}(A_1)$  is the normal Lyapunov exponent on $A_1$ denoted by $L_{\perp, a,b}(1)$.
Simple computations show that, for $b<a/(1-a)$ $, $ $0<a<1/2$, $0<b<1/2$, $\Lambda_{\mathrm{SRB}}(A_1)<0$ and hence $A_1$ is a Milnor (essential) attractor.

For $F_{a,b} \in \mathcal{F}$, we now prove the following statements.
\begin{enumerate}
    \item If  $(a, b) \in \Gamma_0$ then  $A_0 \subset \Phi_0$ has a locally riddled basin $\mathcal{B}(A_0)$.
    \item If  $(a, b) \in \Gamma_1$ then $A_1 \subset \Phi_1$ has a locally riddled basin $\mathcal{B}(A_1)$.
\end{enumerate}
   In particular, for each $(a,b)\in \Gamma$,  both invariant sets $A_0$ and $A_1$ are Milnor essential attractors with  locally riddled basins. 

   We have seen that for $(a,b) \in \Gamma_0$, the invariant set $A_0$ is a Milnor (essential) attractor, and also for $(a,b) \in \Gamma_1$, the invariant set $A_1$ is also a Milnor (essential) attractor.  It remains to show that they possess locally riddled basins.
Our proof relies on statement (1) of Proposition \ref{p1}. We provide a detailed demonstration of the existence of a locally riddled basin for  $A_0$. A similar argument applies to $A_1$.
Consider the invariant Dirac measure $\mu_1$ which is supported at the fixed point  $(0,0)$. Using (\ref{Lambd}), $\Lambda_{\mu_1}(A_0)=\log(1+a)$ which is positive.
By this fact, $0< \Lambda_{\mu_1}(A_0) \leq \Lambda_{\mathrm{max}}(A_0)$, and by above arguments, we have
 $\Lambda_{\mathrm{SRB}}(A_0)<0 < \Lambda_{\mathrm{max}}(A_0)$.
 Now, we claim that for each  $(a,b)\in \Gamma_0$
  there exists $\alpha=\alpha(a,b) > 0$ such that $G_\alpha$ is dense in $A_0$.
It is known that there exists a semi-conjugacy $\pi$ between the shift map $\sigma : \Sigma_2^+ \to \Sigma_2^+$ and the
doubling map $f$  \cite{viana2016foundations}. Leveraging this semi-conjugacy, we establish the assertion.  
Indeed, we define $\rho : \mathbb{I} \to \{1,2\}$ by
\begin{equation}\label{code0}
\rho(x)=
\begin{cases}
1\quad &\text{ if } 0 \leq x<1/2\\
2  \quad &\text{ if } 1/2 \leq x \leq 1.\\
\end{cases}
\end{equation}
For simplicity, set, $g_1\coloneqq g_a$  and $g_2\coloneqq g_b$. Let $w=(\alpha_0, \ldots, \alpha_{n-1})$ be a finite word composed of the digits 1 and 2, and consider the associated cylinder $[\alpha_0, \ldots, \alpha_{n-1}]$.
Take $$g_{a,b,w}^n\coloneqq g_{\alpha_{n-1}}\circ \cdots g_{\alpha_0}.$$
Simple calculations show that $d_{(x,0)}g_{a,b,w}^n=(1+a)^\ell(1-b)^{n-\ell}$, where $\ell$ is the number of digit 1 in the word $w$.
We choose an integer $k$ such that if we take $$w^{\prime}=(\alpha_0, \ldots, \alpha_{n-1}, \underbrace{1, \ldots, 1}_{k\text{-times}}),$$
then $d_{(x,0)}g_{a,b,w^{\prime}}^{n+k} >C$, for some $C>1$. Now, we take the periodic sequence $\omega=(\overline{w^{\prime}})\in \Sigma_2^+$.
Clearly $\omega \in [\alpha_0, \ldots, \alpha_{n-1}]$. Let $x=\pi(\omega)$, then $x$ is a periodic point of $f$ with period $n+k$.
Moreover, $(x,0)\in A_0$ is a periodic point of $F_{a,b}$.
By construction, the set of all such periodic points is dense in $A_0$.
Consider the Dirac measure $\mu_{(x,0)}$ supported at the periodic point $(x,0)$. Since  the derivative $d_{(x,0)}g_{a,b,w^{\prime}}^{n+k} >C$, by taking $\alpha=\log C > 1$, we observe that $\Lambda_{\mu_{(x,0)}}>\alpha$.

We denote by $\mathrm{Per}_\alpha(F_{a,b})$ the set of all periodic points $(x,0)$ with $\Lambda_{\mu_{(x,0)}}>\alpha$. The argument above demonstrates that $\mathrm{Per}_\alpha(F_{a,b})$
 is dense in  $A_0$ which ensures that $G_\alpha$ as defined by (\ref{alpha}) is dense in $A_0$. 

By applying these observations and statement (1) of Proposition  \ref{p1}, $A_0$ is shown to possess a locally riddled basin. 
Similarly, the same holds for $A_1$, thereby the proofs of the first and second statements of Theorem \ref{T1} are completed.

\begin{figure}[!ht]
\begin{center}
   \includegraphics[scale=0.7]{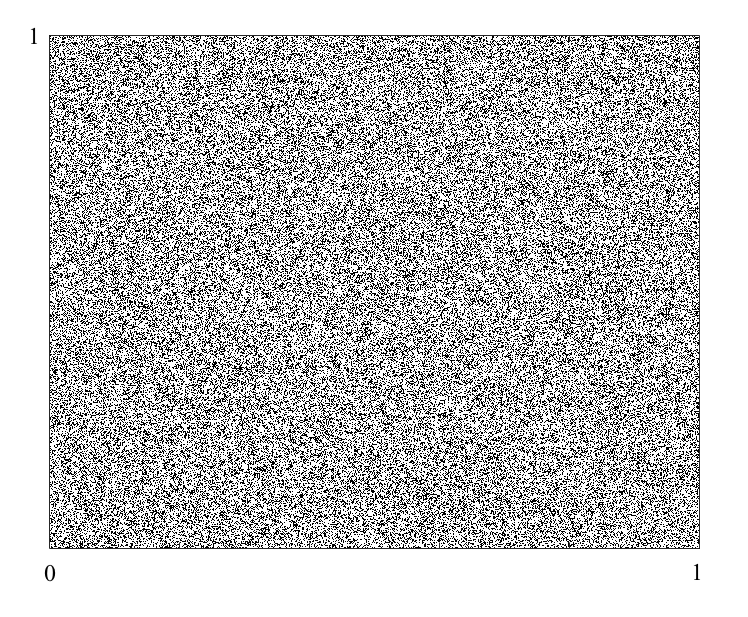}
\end{center}
  \caption{The intermingled basins for the chaotic attractors $A_0$ and $A_1$ are shown for the parameters $a=0.45$ and $b=0.33$ and  with respect to the SRB measure $\nu_{\mathrm{ac}}$. The black points corresponds to the basin of attraction $\mathcal{B}(A_0)$, while the white region 
corresponds to the basin of attraction $\mathcal{B}(A_1)$.   
}
 \label{fig:5}
   \end{figure}

\begin{remark}\label{pos}
   By the above argument, there are two Dirac measures $\mu^i_{(x_i,i)}$, $i=0,1,$ supported on two periodic points $(x_i,i)$ with $\Lambda_{\mu^i_{(x_i,i)}}>1$. So, by Remark \ref{BI}, there are two $f$-invariant measure $\nu^i$, $i=0,1$, with $\lambda_{\nu^i, a, b}(\phi^i)>0$.  
\end{remark}

Now, by Remark \ref{pos} and \cite[Proposition~2.2]{keller2017stability}, the following holds.
\begin{corollary}
Two bounding $\nu_{\mathrm{ac}}$-a.\,e.\@ invariant graphs $\phi^0$ and $\phi^1$ have intermingled basins. 
\end{corollary}
The above corollary ensures that two chaotic attractors $A_i$ have intermingled basins.

To complete the proof of Theorem \ref{T1}, the following lemma is needed.
\begin{lemma}\label{Lemma}
Let $(a,b)\in \Gamma$ and take
\begin{equation*}
 \ell_{a,b}(x)=\sup \{y \in \mathbb{I}: \lim_{n \to \infty} g_{a,b,x}^n(y)=0\}, \quad r_{a,b}(x)=\inf \{y \in \mathbb{I}: \lim_{n \to \infty} g_{a,b,x}^n(y)=1\}.
\end{equation*}
Then $\ell_{a,b}(x)>0$ and $r_{a,b}(x)<1$ for $\nu_{\mathrm{ac}}$-almost all $x \in \mathbb{I}$.
\end{lemma}
\begin{proof}
    We follow the argument used in \cite{bonifant2008schwarzian}.  For any $\epsilon >0$ and $y<\delta$ there exists $\delta >0$  so that 
\[
dg_{a,b,x}(y)\leq dg_{a,b,x}(0)+ \epsilon.
\]
We apply Birkhoff's ergodic theorem for the function $x \to \log (dg_{a,b,x}(0)+ \epsilon)$, for $\nu_{\mathrm{ac}}$-almost all $x$, we have
\begin{equation*}
    \lim_{n \to \infty} \frac{1}{n}\sum_{i=0}^{n-1} \log (dg_{a,b,f^i(x)}(0)+ \epsilon)=1/2 \log((1+a)+ \epsilon)+1/2 \log((1-b)+ \epsilon).
\end{equation*}
But this limit is negative, if $\epsilon$ is small enough, for $(a,b) \in \Gamma$.  So, for $\nu_{\mathrm{ac}}$-almost all $x \in \mathbb{I}$, 
\begin{equation*}
    \sum_{i=0}^{n-1} \log (dg_{a,b,f^i(x)}(0)+ \epsilon) \to -\infty, \quad  \text{as} \quad  n\to \infty,
\end{equation*}
 and 
\begin{equation}
    M(x)=\max \left\{0, \max_{n \geq 1}\sum_{i=0}^{n-1} \log (dg_{a,b,f^i(x)}(0)+ \epsilon) \right\} 
\end{equation}
exists.  Take $y_0 < \delta\e^{-M(x)} \leq  \delta$.  Then $y_n = g_{a,b,x}^n(y_0)$ satisfies 
$$y_n <\e^{\sum_{i=0}^{n-1} \log (dg_{a,b,f^i(x)}(0)+ \epsilon) }\e^{-M(x)}\delta  \leq  \delta ,$$
for all $n \geq 0$, so $\lim_{n \to \infty} y_n =0$.  Therefore, $\ell_{a,b}(x)>0$, for $\nu_{\mathrm{ac}}$-almost all $x \in \mathbb{I}$. A similar argument shows that $r_{a,b}(x)<1$ for $\nu_{\mathrm{ac}}$-almost all $x \in \mathbb{I}.$
This completes the proof of the lemma. 
\end{proof}
Clearly both sets $\ell_{a,b}$ and $r_{a,b}$ are invariant graphs.
By the above lemma,  since $\lambda_{a,b, \nu_{\mathrm{ac}}}(\phi^0)<0$ and $\lambda_{a,b, \nu_{\mathrm{ac}}}(\phi^1)<0$, and by \cite[Proposition~1.6]{keller2017stability},  we conclude that $\ell_{a,b}=r_{a,b}$ almost everywhere with respect to $\nu_{\mathrm{ac}}$. We denote this invariant graph by $\phi^*$. 
This proves the last statement of Theorem \ref{T1}.

\begin{figure}[!ht]
\begin{center}
\includegraphics[scale=0.7]{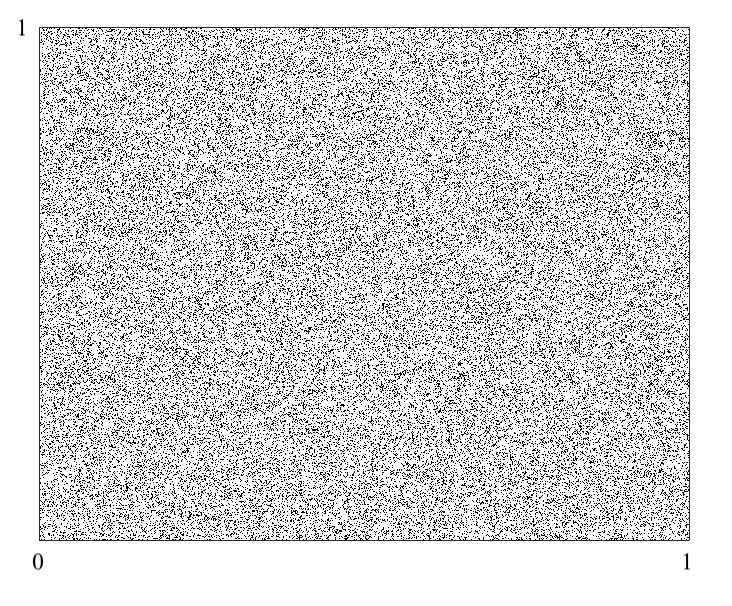}
\end{center}
  \caption{The intermingled basins for the chaotic attractors $A_0$ and $A_1$ are shown for the parameters $a=0.45$ and $b=0.33$, but in this case, we consider a Gibbs measure $\nu$ that assigns weights of 1/3 and 2/3 to the intervals  [0,1/2] and [1/2,1], respectively.  The black points corresponds to the basin of attraction $\mathcal{B}(A_0)$, while the white region corresponds to the basin of attraction $\mathcal{B}(A_1)$. }
 \label{fig:55}
   \end{figure}
%%%%%%%%%%%%%%%%%%%%
\subsection{Proof of Theorem \ref{T2}}
Let $\mu_1$ be the invariant Dirac measure supported at the fixed point $(1,0)$. Using (\ref{Lambd}), $\Lambda_{\mu_1}(A_0)=\log(1-b)$ which is negative. By this fact, $\lambda_{\min}(A_0)\leq \Lambda_{\mu_1}(A_0) < 0$.
Also, it is easy to see that for $0<b<a/(1+a)$ , $\Lambda_{\mathrm{SRB}}(A_0)>0$.
By these facts, we get $\lambda_{\min}(A_0)<0<\Lambda_{\mathrm{SRB}}(A_0)$. Since $\mu_{\mathrm{SRB}}(A_0)$ is equivalent to the Lebesgue measure and whose support is $A_0$,
 $m(\bigcup_{\mu \neq \mu_{\mathrm{SRB}}} G_\mu)=0$, where $m$ denotes the Lebesgue measure on $\Phi_0$. Clearly, $\mu_{\mathrm{SRB}}(A_0)$-almost all Lyapunov exponents are non-zero. Using these facts and statement (2) of Proposition \ref{p1}, $A_0$ is shown to be a chaotic saddle, thus confirming statement (1) of Theorem \ref{T2}.

A similar argument shows that $A_1$ is also a chaotic saddle, if $a/(1-a)<b<1/2$ and $0<a<1/2$.
It suffices to take $\nu_1$ as the invariant Dirac measure supported at the fixed point  $(0,1)$. By (\ref{1Lambd}), $\Lambda_{\nu_1}(A_1)=\log(1-a)$ which is negative.
 By this fact, $\lambda_{\min}(A_1)\leq \Lambda_{\nu_2}(A_1) < 0$.
Additionally, it is evident that for $a/(1-a)<b<1/2$ and $0<a<1/2$,  $\Lambda_{\mathrm{SRB}}(A_1)>0$.
These observations imply that $\lambda_{\min}(A_1)<0<\Lambda_{\mathrm{SRB}}(A_1)$. Furthermore, since $\mu_{\mathrm{SRB}}(A_1)$ is equivalent to the Lebesgue measure and its support is $A_1$, we have
 $m(\bigcup_{\mu \neq \mu_{\mathrm{SRB}}(A_1)}G_\mu)=0$. Based on these facts  and statement (2) of Proposition \ref{p1}, for  $a/(1-a)<b<1/2$ and $0<a<1/2$, $A_1$ is a chaotic saddle,  verifying statement $(2)$ of Theorem \ref{T2}.

\begin{figure}[!ht]
\begin{center}
    $\begin{array}{cc}
  \includegraphics[scale=0.8]{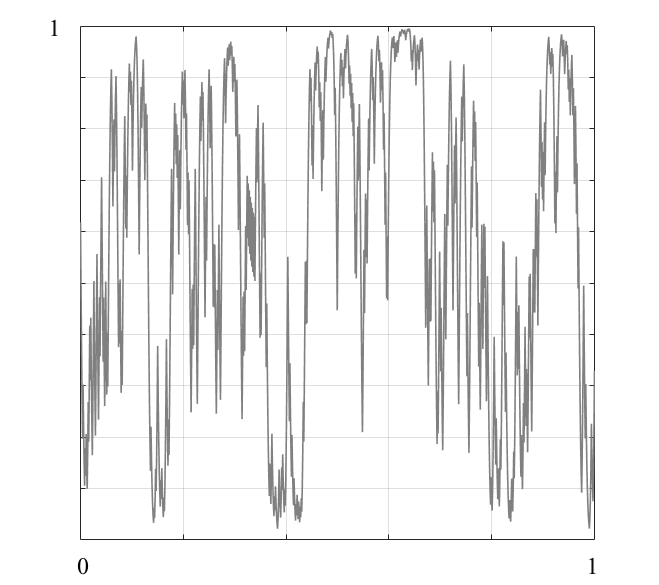}
  \end{array}$
\end{center}
  \caption{The separating graph $\phi^*$ for  the parameters $a=0.49$, $b=0.49$.}
 \label{fig:6}
   \end{figure}

\section{A thermodynamic Loynes exponent and stability index
}\label{section4}
In this section, we prove Theorem \ref{T3} and Theorem \ref{T4}.  
Before we present the necessity distortion  estimates, let us recall the following lemma from \cite{pesin1997multifractal}. 

\begin{lemma}\label{perl}
    Let $f : \mathbb{I} \to \mathbb{I}$ be a Hölder continuous piecewise expanding map, and let $\psi:\mathbb{I} \to \mathbb{R}$ be a Hölder continuous potential. Then there exists a unique equilibrium state $\nu_{\psi}$ with respect to the potential function $\psi$.  Furthermore, the following are true: 
    For Hölder continuous potentials $\psi$ and $\eta$, and $t \in \mathbb{R}$, the map $t \mapsto P(\psi + t \eta)$ is real analytic. Additionally, we have 
    \begin{equation*}
        \frac{\partial}{\partial t} P(\psi + t \eta)|_{t=0} =\int \eta \d\nu_{\psi}, \quad \text{and} \quad \frac{\partial^2}{\partial^2 t} P(\psi + t \eta)|_{t=0}\geq 0.
    \end{equation*}
Furthermore, the second derivative at t is zero if and only if $\psi + t \eta$ is cohomologous to a constant.
Moreover, for a given (piecewise) Hölder potential function $\psi$, each equilibrium state   with respect to $\psi$ is a Gibbs measure with respect to the same potential function $\psi$, and conversely.  
\end{lemma}

 To prove Theorem \ref{T3} and Theorem \ref{T4}, we closely follow the approaches outlined in  \cite{keller2014stability} and \cite{keller2017stability}. First, we present some preliminaries concepts and a few auxiliary lemmas to provide necessary distortion estimates.
 
Let the base map $f$  be a piecewise expanding and piecewise $C^{1+\alpha}$-Hölder mixing Markov map with two branches. 
Assume $\mathcal{V}_n(x)$ is the family of all interval neighborhoods $V$ of $x \in \mathbb{I}$ such that $f^n|_V : V \to f^n V$ is a diffeomorphism. 
So, there is a distortion constant $D \geq 1$ such that for all $n > 0$, all $x \in \mathbb{I}$, all $V \in \mathcal{V}_n(x)$ and all $\tilde{x} \in V$
\begin{equation}\label{7}
 \e^{-D} \leq \left| \frac{(df^n)(\tilde{x})}{(df^n)(x)} \right|, \;
\left| \frac{dg^n_{a,b,\tilde{x}}(\phi^i)}{dg_{a,b,x}^n(\phi^i)} \right| \leq\e^D,
  \end{equation}

where $i=0,1$. Moreover, the following holds: 
\begin{equation}\label{8}
\left| \log \left(  \frac{dg^n_{a,b,\tilde{x}}(\phi^i)}{dg_{a,b,x}^n(\phi^i)} \right) \right| \leq D \cdot \left( |f^n(\tilde{x}) - f^n(x)| + \sum_{i=0}^{n-1} |g_{a,b,\tilde{x}}^i(\tilde{y}) - g^i_{a,b,x}(y)|^\alpha \right),
\end{equation}
for all $y, \tilde{y} \in \mathbb{I}$ (see \cite[Lemma~2.6]{ott2012memory}  and \cite[Section~3]{keller2017stability}). 

The following lemma directly follows from equation \eqref{8}.

\begin{lemma}\label{L3.2}
If \( \lim_{n \to \infty} |g^n_{a,b,x}(\tilde{y})- g^n_{a,b,x}(y)| = 0 \), then, for each \( \delta > 0 \), there exists \( C = C(y, \tilde{y}, x, \delta) > 0 \) such that\
\[
|\tilde{y} - y| dg^n_{a,b,x}(y)\e^{-C - n\delta} 
\leq |g^n_{a,b,x}(\tilde{y}) - g^n_{a,b,x}(y)| 
\leq |\tilde{y} - y| dg^n_{a,b,x}(y)\e^{C + n\delta},
\]
for all \( n \in \mathbb{N} \).
\end{lemma}
Another consequence is the following lemma which is obtained as in \cite[Lemma~3.3]{keller2017stability}:

\begin{lemma}\label{L3.3}
 Let \( n > 0 \), \( x \in \mathbb{I} \), and \( V \in \mathcal{V}_n(x) \). Then
\begin{equation}\label{9}
\e^{-D \cdot |f^n(\tilde{x}) - f^n(x)|} 
\leq 
\frac{g^n_{a,b,\tilde{x}}(\tilde{y}) - g^n_{a,b,\tilde{x}}(y)}{g^n_{a,b,x}(\tilde{y}) - g^n_{a,b,x}(y)} 
\leq 
\e^{D \cdot |f^n(\tilde{x}) - f^n(x)|}
\end{equation}
for all \( y, \tilde{y} \in \mathbb{I}\) and \( \tilde{x} \in V \), and if \( u : f^n(V) \to \mathbb{I} \) is such that 
\[
\{ (z, u(z)) : z \in f^n(V) \} = F_{a,b}^n(V \times \{y\}),
\]
i.\,e., \( u(f^n(\tilde{x})) = g^n_{a,b,\tilde{x}}(y) \), then
\[
\left| \log \frac{u(z)-\phi^0}{u(f^n(x))-\phi^0} \right| \leq D \cdot |z - f^n(x)|
\]
for all \( z \in f^n(V) \) (and similarly for \( \phi^1 \)).
\end{lemma}
The following lemma is a counterpart of \cite[Lemma~3.4]{keller2017stability} to our setting which uses the Hölder property as defined in Remark \ref{Holder}.
\begin{lemma}\label{L3.4}
Let \( \delta > 0 \). There exist \( n_0 \in \mathbb{N} \) and \( \delta_0 \in (0, \delta) \), which depend only on \( \delta \), with the following property: For all \( x \in \mathbb{I} \), \( y \in \mathbb{I} \), \( \vartheta \in (-\delta_0, \delta_0) \), and \( n > n_0 \) such that
\[
|\vartheta| \cdot dg^k_{a,b,x}(y) \leq\e^{-2k\delta} \quad \text{for } k = n_0, \dots, n,
\]
it holds that:
\begin{equation}\label{10}
    |\vartheta| \cdot dg^n_{a,b,x}(y) \cdot\e^{-3\delta n} 
\leq 
|g^n_{a,b,x}(y + \vartheta) - g^n_{a,b,x}(y)| 
\leq 
|\vartheta| \cdot dg^n_{a,b,x}(y) \cdot\e^{\delta n}.
\end{equation}

\end{lemma}
\subsection{Proof of Theorem \ref{T3}}

We closely follow the arguments presented in \cite[Sections 4.1 and 7]{keller2014stability} and \cite[Section~3.2]{keller2017stability}, which provide foundational insights into the framework.
Consider the Hölder continuous potential $\psi$ with $P(\psi)=0$ and corresponding Gibbs measure $\nu=\nu_\psi$ as in condition (H1). Using \cite{ruelle2004thermodynamic, viana2016foundations}, we get the following result.
\begin{lemma}\label{L00}
    The measure $\nu$ is positive on any open set and has full support.
\end{lemma}
\begin{proof}
    For the Hölder continuous potential \(\psi\), 
the Perron-Frobenius (transfer) operator \(\mathcal{L}_\psi\) is defined 
for functions \(k: [0,1] \to \mathbb{R}\) by:
\[
\mathcal{L}_\psi k (x) = \sum_{y \in f^{-1}(x)}\e^{\psi(y)} k(y).
\]
It is clear from the definition that \(\mathcal{L}_\psi\) is a positive operator:  
if \( k(x) \geq 0 \) for every \( x \in \mathbb{I}\), then \( \mathcal{L}_\psi k(x) \geq 0 \) for every \( x \in \mathbb{I} \).  
It is also easy to check that \(\mathcal{L}_\psi\) is a continuous operator. 

Then, the dual of the transfer operator \(\mathcal{L}_\psi\) is the linear operator 
\(\mathcal{L}^*_\psi : \mathcal{M}(\mathbb{I}) \to \mathcal{M}(\mathbb{I})\) defined by
\[
\int k \d(\mathcal{L}^*_\psi \mu) = \int \mathcal{L}_\psi h \d\mu
\]
for every continuous function \( k \) and \( \mu \in \mathcal{M}(\mathbb{I}) \). 
Consider the spectral radius \( \lambda = \varrho(\mathcal{L}^*_\psi) = \varrho(\mathcal{L}_\psi) \).  
Then there exists some probability measure \( \mu \) on \( \mathbb{I}\) such that  
\begin{equation}\label{P1}
    \mathcal{L}^*_\psi \mu = \lambda \mu.
\end{equation}
The measure $\mu$ is called a reference measure \cite[Section~12]{viana2016foundations}. The doubling map \( f \) admits a Jacobian with respect to \( \mu \), given by  
\[
J_{\mu} f = \lambda\e^{-\psi},
\]
(see \cite[Lemma~12.1.3]{viana2016foundations}. We show that $\mu$ is supported on the whole of $\mathbb{I}$. Indeed, suppose, by contradiction, that there exists some open set $U$ such that $\mu(U)=0$, and we may assume that $U$ is contained in one of  Markov interval $I_i$, $i=0,1$. Note that $f$ is an open map, since it is a local diffeomorphism. Thus, the image $f(U)$ is also an open set. Moreover, we may write $U$ as a finite disjoint union of domains of invertibility $A$. For each one of them, 
\[
\mu(f(A)) = \int_A J_{\mu} f \d\nu = 0.
\]
Therefore,  
\[
\mu(f(U)) = 0.
\]
By induction, it follows that  
\[
\mu(f^n(U)) = 0 \quad \text{for every } n \geq 0.
\]
Since  \( f \) is topologically exact, there exists \( n \geq 1 \) such that  
\[
f^n(U) = \mathbb{I}.
\]
This contradicts the fact that \( \mu(\mathbb{I}) = 1 \).
Hence, $\mu $ is supported on the whole interval $\mathbb{I}$. By the result \cite[Lemma~12.1.11]{viana2016foundations}, the transfer operator $\mathcal{L}_\psi$ admits some positive eigenfunction $h$ associated with the eigenvalue $\lambda$ satisfying \( \mathcal{L}_\psi h = \lambda h \). Moreover,  
\[
\int h\d\mu = 1.
\]
Then, the equilibrium Gibbs measure $\nu=\nu_\psi$ satisfies $\nu=h \mu$.
Moreover, by  \cite[Lemma~12.1.12]{viana2016foundations},  \( \nu \) is invariant under \( f \), it admits a Jacobian with respect to \( \nu \), given by  
\[
J_{\nu} f = \lambda\e^{-\psi}{\frac{(h \circ f)}{h}}.
\]
In particular, $\nu$ is equivalent to the reference measure $\mu$. This fact, together with \cite[Lemm~12.1.4]{viana2016foundations} Lemma 12.1.4, gives that $\text{supp}~\nu = \mathbb{I}$ and it is positive on open sets.

\end{proof}

\begin{corollary}\label{C00}
   Let $R$ be a Markov interval of $f$ and $R_0 \subset R$ with $\nu_\psi(R \setminus R_0)=0$. Then there is $k \in \mathbb{N}$ such that $\nu_\psi(\mathbb{I} \setminus f^k(R_0)) = 0$. 
\end{corollary}

\begin{proof}
From equation (\ref{P1}) and the assumption that $P(psi) = 0$, we conclude for the eigenmeasure $\mu$ of the dual transfer operator with respect to the leading eigenvalue $\lambda=\e^{P(\psi)}=1$,  
\begin{equation}\label{P2}
 \mathcal{L}_{\psi}^* \mu =  \mu.   
\end{equation}
Take a measurable set $B \subset R \subset I_i$, $i=0,1,$ with $\nu_\psi(B)=0$. 
Since $\nu_\psi$ is equivalent to $\mu$ we have $\mu(B)=0$. By virtue of \cite[Lemma~1]{keller1998equilibrium} and (\ref{P2}), we have $\mu(f(B))=0$ and by the equivalence of $\nu_\psi$ and $\mu$, $\nu_\psi(f(B))=0$ also holds.
As $f$ is a mixing Markov map,  there is $k \in \mathbb{N}$ such that $f^k(R)=\mathbb{I}$. Taking $B=R_0^c=R \setminus R_0$, we find inductively $\nu_\psi (f^k(R_0^c)) = 0$, and hence, $\nu_\psi(\mathbb{I} \setminus f^k(R_0)) = 0$.  
\end{proof}

For \( t \in \mathbb{R} \), denote by \( \mathcal{L}_t^i \), $i=0,1,$ the transfer operators
\begin{equation}\label{12}
\mathcal{L}_t^i: L^1_{\nu_\psi}(\mathbb{I}) \to L^1_{\nu_\psi}(\mathbb{I}), \quad \mathcal{L}_t^i h(x) = \sum_{\tilde{x} \in f^{-1}(x)} 
h(\tilde{x})\e^{\psi(\tilde{x})}\e^{t \log dg_{a,b,\tilde{x}}(\phi^i)},
\end{equation}
and let \( \varrho(\mathcal{L}_t^i) \) be its spectral radius. Then \( p_{i,\psi}(t) = \log \varrho(\mathcal{L}_t^i) \), and this is a strictly convex differentiable function of \( t \), see e.\,g.\@ \cite{parry1990zeta}.
This is a consequence of the Gibbs property of $\nu_\psi$, expansiveness of $f$ and since the pressure function $p_{i,\psi}$ is is strictly convex and differentiable, allowing fine control over the spectral radius \( \varrho(\mathcal{L}_t^i) \) . 
As $\psi$ is normalized, we have $P(\psi) = 0$. After possibly adding a coboundary to $\psi$, we can assume that 
$\mathcal{L}^i_t 1 = 1$, where $1$ denotes the constant function. Moreover, we must have $\psi(x) < 0$ for all $x \in \mathbb{I}$. 
It is well known that the Banach space $L^1_{\nu_\psi}(\mathbb{I})$ contains the constants, such that 
$\mathcal{L}^i_t$ has $\exp{{p_{i,\psi}(t)}}$ as a simple maximal eigenvalue and with the remainder of the spectrum contained within a disc of radius $\gamma_t < \exp{p_{i,\psi}(t)}$.  
In particular, we can write 
\[
(\mathcal{L}^{i}_t)^n = \exp({np_{i,\psi}(t)}) \pi_t + O(\gamma_t^n)
\]
where $\pi_t$ is a projection operator corresponding to the maximal eigenvalue, and $O(\gamma_t^n)$ represents the decay due to the rest of the spectrum. 

We will prove only the identity where  the limit is given by \( t_0^* \);  the other identity can be established using a similar argument. 
We will present the following large deviations theorem, first due to Plachky and Steinebach \cite{plachky1975theorem}, adapted to our situation.  Note that, as before  the derivative of any function $h$, is denoted by $d h$.
\begin{proposition}\label{P4.7}
 Let $(t_-, t_+)$ be an open interval containing $t_0^*$ and suppose that, for $t \in (t_-, t_+)$, $p_{0,\psi}(t)$ is a differentiable 
function with $d p_{0,\psi}(t)$ strictly monotone. Suppose that:
\begin{itemize}
    \item[(i)] $\int\e^{t \log dg^n_{a,b,x}(\phi^0)} \d\nu < 1$ for all $t \in [0, t_-)$,
    \item[(ii)] for all $t \in (t_-, t_+)$, we have
    \[
    \lim_{n \to \infty} \frac{1}{n} \int \e^{t d g_{a,b,x}^n(\phi^0)} \d\nu = p_{0,\psi}(t).
    \]    
\end{itemize}
Then
\[
\lim_{n \to \infty} \frac{1}{n} \log \nu\left(\left\{x \in \mathbb{I} \mid \log  dg_{a,b,x}^n(\phi^0) > n d p_{0,\psi}(t_0^*)\right\}\right) 
= p_{0,\psi}(t_0^*) - t_0^* d p_{0,\psi}(t_0^*).
\]

\end{proposition}
Now, we show that the assumptions $(i)$ and $(ii)$ of the previous proposition are satisfied for our setting.

First, it is well known that $p_{0,\psi}$ is a convex analytic function.

\begin{lemma}
There exists a unique $t_0^* > 0$ such that $p_{0,\psi}(t_0^*) = 0$. Moreover, $dp_{0,\psi}(t_0^*) > 0$ and $dp_{0,\psi}(t)$ is strictly increasing on an open interval $(t_-, t_+)$ that contains $t_0^*$.
 
\end{lemma}
\begin{proof}
  Recall from Lemma \ref{perl} that if $\psi$ is Hölder
 continuous, $P(\psi) = 0$, has an equilibrium state $\nu=\nu_\psi$, and $\eta$ is Hölder continuous, then
\[
\left.\frac{\partial P(\psi + t\eta)}{\partial t}\right|_{t=0} = \int \eta \d\nu.
\]
Moreover,
\[
\left.\frac{\partial^2 P(\psi + t \eta)}{\partial t^2}\right|_{t=0} \geq 0,
\]
with equality if and only if $\eta$ is cohomologous to a constant.
Assume condition (H1)  holds. Take the potential $\psi$ with equilibrium $\nu=\nu_\psi.$ First, note that $p_{0,\psi}(0) = 0$ since $\psi$ is normalized. By the above, we have $d p_{0,\psi}(0) = \int \log dg_{a,b,x}(\phi^0) \d\nu < 0$. Since $\int \log dg_{a,b,x}(\phi^0) \d\nu^0 > 0$, it follows that $\log dg_{a,b,x}(\phi^0)$ cannot be cohomologous to a constant, where the meaure $\nu^0$ is given by Remark \ref{pos}, see also the proof of Theorem \ref{T1}. Hence, $p_{0,\psi}(t)$ is strictly convex.
By the variational principle,
\[
p_{0,\psi}(t) = \sup \left\{ h_\mu(f) + \int \psi\d\mu + t \int \log dg_{a,b,x}(\phi^0)\d\mu:\mu \in \mathcal{M}_f \right\}.
\]
 Hence,
\[
p_{0,\psi}(t) \geq h_{\nu^0}(f) + \int \psi \d\nu^0 + t \int \log dg_{a,b,x}(\phi^0) \d\nu^0.
\]
It follows that $p_{0,\psi}(t) \to \infty$ as $t \to \infty$, since $\int \log dg_{a,b,x}(\phi^0) \d\nu^0 > 0$.
Since $p_{0,\psi}(t)$ is analytic and convex, there exists a unique $t_0^* > 0$ such that $p_{0,\psi}(t_0^*) = 0$. Moreover, $dp_{0,\psi}(t_0^*) > 0$. Therefore, there is an interval $(t_-, t_+)$ containing $t_0^*$ on which $dp_{0,\psi}(t) > 0$. As $p_{0,\psi}$ is convex, $dp_{0,\psi}$ is non-decreasing.
To show that $dp_{0,\psi}$ is strictly increasing on $(t_-, t_+)$, assume for contradiction that $d^2p_{0,\psi}(t) = 0$ on a subinterval of $(t_-, t_+)$. Then $d^2p_{0,\psi}(t) = 0$ for all $t$ by analytic continuation, which implies that $dp_{0,\psi}(t)$ is constant for all $t$. This contradicts $dp_{0,\psi}(0) < 0$ and $dp_{0,\psi}(t_0^*) > 0$.

\end{proof}

We now confirm that our setting satisfies the hypotheses of Proposition \ref{P4.7}. As $p_{0,\psi}(t)$ is convex and not linear, $dp_{0,\psi}(t)$ is strictly increasing. Hypothesis ($i$) of Proposition \ref{P4.7} holds by definition, as $dg_{a,b}$ is piecewise continuous and therefore bounded. The only remaining task is to verify the convergence in ($ii$). 
To see this, observe that
\begin{align*}
\lim_{n \to \infty} \frac{1}{n} \log \int \e^{t dg_{a,b,x}^n(\phi^0)} \d\nu 
& = \lim_{n \to \infty} \frac{1}{n} \log \int (\mathcal{L}^0_t)^n 1 \d\nu 
\\
&= \lim_{n \to \infty} \frac{1}{n} \log \int \e^{n p_{0,\psi}(t)} \pi_t 1 + O(\gamma_t^n) \d\nu = p_{0,\psi}(t).
\end{align*}
We can now apply Proposition \ref{P4.7} to complete the proof of Theorem \ref{T3}.
Indeed, fix any \( t \in (0, t_0^*) \) and choose \( \delta > 0 \) such that \( \varrho(\mathcal{L}_t^0) \e^{4t\delta} < 1 \). There is a constant \( C = C_{t, \delta} > 0 \) such that
\begin{equation}\label{13}
\left\| \left(\mathcal{L}_t^0\right)^n 1 \right\|_1 \leq C \left( \varrho\left(\mathcal{L}_t^0\right) \e^{t\delta} \right)^n \leq C \e^{-3n\delta} \quad \text{for all } n \geq 1.
\end{equation}
Using equation (\ref{13}) and the reasoning outlined in \cite[Lemma~3.6]{keller2017stability}, we derive the following lemma.
\begin{lemma}\label{L3.6}
    Let \( t \in (0, t_0^*) \) and \( \delta > 0 \) be as chosen above. Then
\[
\nu_\psi \left( \left\{ x\in \mathbb{I}: \vartheta \cdot dg^n_{a,b,x}(\phi^0) > \e^{-2n\delta} \right\} \right) 
\leq C \e^{-n t \delta} \vartheta^t
\]
for all \( \vartheta> 0 \) and \( n \geq 1 \).

\end{lemma}
To prove Theorem \ref{T3}, we first  show that 
\begin{equation}\label{low}
\limsup_{\varepsilon \to 0} \frac{\log \nu_\psi\left( \{ x \in \mathbb{I} : \phi^*(x) < \phi^0 + \varepsilon \} \right)}{\log \varepsilon} > t \quad \text{for each } t \in (0, t_0^*).
\end{equation}
Let \( n_0 \) and \( \delta_0 \) be as in Lemma \ref{L3.4}, and consider \( \varepsilon \in (0, \delta_0) \). Then
\[
\nu_{\psi}\left( \{x\in \mathbb{I} : \phi^*(x) < \phi^0 + \varepsilon \} \right) 
\leq \nu_\psi\left( \{ x\in \mathbb{I} : \exists n > n_0 \text{ s.t. } g^n_{a,b,x}(\phi^0 + \varepsilon) > \phi^0 + \e^{-\delta n} \} \right)
\]
\[
\leq \nu_\psi\left( \{ x\in \mathbb{I} : \exists k > n_0 \text{ s.t. } \varepsilon \cdot dg^k_{a,b,x}(\phi^0) > \e^{-2k\delta} \} \right),
\]
in view of the upper estimate in Lemma \ref{L3.4}. Hence, by Lemma \ref{L3.6},  
\[
\nu_\psi\left( \{ x\in \mathbb{I}: \phi^*(x) < \phi^0+ \varepsilon \} \right) \leq C_{t, \delta_0} \varepsilon^t.
\]
By taking logarithms and dividing by 
\( -\log \varepsilon \), we obtain the lower bound  in (\ref{low}).
Note that, by Lemma \ref{L00},  $\nu=\nu_{\psi}$ is fully supported on $\mathbb{I}$. Since  $\nu_\psi$ assigns positive measure to all open subsets and using \cite[Lemma~3.7]{keller2017stability}, we arrive at the following lemma,  establishing the upper bound.
\begin{lemma}\label{L3.7}
 For all \( \beta > 0 \) and \( \ell > 0 \), there exist \( \gamma = \gamma(\beta) > 0 \) and \( n_0 = n_0(\ell) \in \mathbb{N} \) such that for all \( z \in \mathbb{I} \), each interval \( I \subseteq \mathbb{I} \) of length at least \( \ell \), and all \( n > n_0 \), the following holds:
\[
\nu_\psi \left( \{ x \in I: \phi^*(x) < z \} \right) 
\geq \gamma \cdot \nu_\psi\left( \{ x \in \tilde{I} : g^n_{a,b,x}(z) - \phi^0 > \beta \} \right),
\]
where \( \tilde{I} \) denotes the middle third of  \(I\).
Indeed, one can choose 
\[
\gamma \coloneqq \e^{-D} \cdot \min_{R \in \mathcal{R}} \nu_\psi\left( \{ x \in R: \phi^*(x) < \phi^0 + \beta \e^{-D} \} \right),
\]
where the minimum extends over the family \( \mathcal{R} \) of all Markov intervals of \( f \).

\end{lemma}
\begin{proof}
    Fix $z \in \mathbb{I}$, $n > 1$ and an interval $I \subset \mathbb{I}$, and denote by $\mathcal{U}$ the family of all maximal monotonicity intervals $U \subseteq \mathbb{I}$ of $f^n$ which contain a point $x_U$ such that
\begin{equation*}
    g^n_{a,b,x_U}(z) - g^n_{a,b,x_U}(\phi^0) = g^n_{a,b,x_U}(z) - \phi^0  \geq \beta.
\end{equation*}
Such an interval need not exist for each $\beta > 0$.  Denote the inverse of $f^n|_U$ by $h = h_U : f^n(U) \to U$. Then, for each $U \in \mathcal{U}$,
\begin{align*}
   & \nu_\psi \{x \in f^n(U) : \phi^*(x) - \phi^0 < \beta \e^{-D} \}\\
    &\leq \nu_\psi \{x \in f^n(U) : \phi^*(x) - \phi^0   < (g^n_{a,b,x_U}(z) - g^n_{a,b,x_U}(\phi^0)) \e^{-D} \} \\
    &\leq \nu_\psi \{x \in f^n(U) : g^n_{a,b,h(x)}(\phi^*(h(x))) - g^n_{a,b,h(x)}(\phi^0) < g^n_{a,b,h(x)}(z) -g^n_{a,b,h(x)}(\phi^0)  \} \\
    &= \nu_\psi \{x \in f^n(U) : \phi^*(h(x)) < z \}.
\end{align*}
where we applied (\ref{9}) for the second inequality and used the monotonicity of $g^n_{a,b,h(x)}$ for the last one. Using the distortion bound (\ref{7}), this implies
\[
\frac{1}{\nu_\psi(f^n(U))} \nu_\psi \{x \in f^n(U) : \phi^*(x) - \phi^{0} < \beta \e^{-D} \} \leq \e^D \cdot \frac{1}{\nu_\psi(U)} \nu_\psi \{x \in U : \phi^*(x) < z \}.
\]
Note that, by Lemma \ref{L00}, $\nu_\psi$ assigns positive measure to all open subsets. Therefore, 
\[
\sum_{U \in \mathcal{U}} \nu_\psi \{ x \in U : \phi^*(x) < z \}\geq\e^{-D} \cdot \sum_{U \in \mathcal{U}} \frac{\nu_\psi(U)}{\nu_\psi(f^n(U))} \{x \in f^n(U): \phi^*(x) - \phi^{0} < \beta\e^{-D} \}
\geq \gamma \cdot \sum_{U \in \mathcal{U}} \nu_\psi(U)
\]
For the last inequality, we used that $f^n(U)$ contains at least one Markov interval when $n > n_0$
and $n_0$ is sufficiently large. Choosing $n_0 = n_0(\ell)$ even larger, if necessary, the at most two $U \in \mathcal{U}$ which are not fully contained in $I$ are disjoint to $\tilde{I}$. Hence,
\[
\nu_\psi \{x \in I : \phi^*(x) < z \} \geq \gamma \cdot \sum_{U \in \mathcal{U}} \nu_\psi(U) \geq \gamma \cdot\nu_\psi \{ x \in I~ : g^n_{a,b,x}( z) - \phi^{0} \geq \beta \}.
\]
It is enough to prove that $\gamma$ is strictly positive: Let $\vartheta\coloneqq\e^{-D}\beta$ and suppose for a contradiction that there is some Markov interval $R$ of $f$ such that $\phi^*(x) \geq \phi^{0} + \vartheta$ for all $x$ in a full measure subset $R_0$ of $R$.
By Corollary \ref{C00}, there is $k \in \mathbb{N}$ such that $\nu_\psi (\mathbb{I} \setminus f^k(R_0))=0$.
Hence, by the claim, for $\nu_\psi$-$a.e$. $\tilde{x} \in \mathbb{I}$ there is $x^\prime\in R_0$ such that $\phi^*(\tilde{x}) = g^k_{a,b,x^{\prime}}(\phi^*(x^{\prime})) > \inf_{x \in \mathbb{I}} g^k_{a,b,x}(\vartheta) > \phi^{0}$, which is incompatible with \cite[Proposition~ 2.2]{keller2017stability}.
By one sided, by Koebe's Principle, the following holds (see \cite{keller2017stability}):
\begin{equation}\label{14}
\frac{g^n_{a,b,x}(z) - \phi^0}{z - \phi^0} 
> 
\frac{\phi^1 - g^n_{a,b,x}(z)}{\phi^1 - z} \cdot dg^n_{a,b,x}(\phi^0)
\end{equation}
for all \( x \in \mathbb{I} \), \( z \in \mathbb{I} \), and \( n \geq 1 \).
Now, we prove the following inequality:
\begin{equation}\label{high}
    \limsup_{\varepsilon \to 0} \frac{\log \nu_\psi\left( \{ \phi^* < \phi^0 + \varepsilon \} \right)}{\log \varepsilon} \leq t_0^*.
\end{equation}
To facilitate later use, we prove a slightly stronger statement, namely:  
If \( (I_\varepsilon)_{\varepsilon > 0} \) is any family of intervals with \( \ell \coloneqq \inf_\varepsilon |I_\varepsilon| > 0 \), then
\begin{equation}\label{15}
\limsup_{\varepsilon \to 0} \frac{\log \nu_\psi\left( \{ x\in I_\varepsilon : \log \phi^*(x) < \phi^0+ \varepsilon \} \right)}{\log \varepsilon} \leq t_0^*. 
\end{equation}
Fix any \( \beta \in (0, 1) \), e.\,g., \( \beta = \frac{1}{2} \). If \( x \in \mathbb{I}\), \( z \in (0,1/2) \subset \mathbb{I} \), and \( n \geq 1 \) are such that 
\[
dg^n_{a,b,x}(\phi^0) > z^{-2\beta} \phi^0,
\]
then \( g^n_{a,b,x}(z) - \phi^0 > \beta \). Indeed, otherwise \( \phi^1 - g^n_{a,b,x}(z) > 2 - \beta \), so that (\ref{14}) implies 
\[
dg^n_{a,b,x}(\phi^0) < \frac{ \beta}{2- \beta} \cdot \frac{\phi^1-z}{z-\phi^0} < \frac{2\beta}{z-\phi^0}.
\]
Hence, Lemma \ref{L3.7} yields for each \( z = \phi^0+ \varepsilon \) with \( \varepsilon \in (0, 1/2) \) and \( n \geq n_0(\ell) \),
\[
\nu_\psi\left( \{ x \in I_\varepsilon : \phi^*(x) < z \} \right) 
\geq \gamma \cdot \nu_\psi\left( \{ x \in \tilde{I}_\varepsilon : dg^n_{a,b,x}(\phi^0) > z^{-2\beta} \phi^0 \} \right),
\]
so that, for any choice of \( n_\varepsilon \geq n_0(\ell) \),
\[
\limsup_{\varepsilon \to 0} \frac{\log \nu_\psi\left( \{ x\in I_\varepsilon : \phi^*(x) < \phi^0 + \varepsilon \} \right)}{\log \varepsilon} 
\leq \limsup_{\varepsilon \to 0} \frac{\log \nu_\psi\left( \{ x \in \tilde{I}_\varepsilon : dg^{n_\varepsilon}_{a,b,x}(\phi^0) > 2 \varepsilon^\beta \} \right)}{\log \varepsilon}.
\]
Exactly as in \cite[Equations (4.9) and (4.10)]{keller2014stability}, let  \( \alpha = (p_{0,\nu_\psi})'(t_0^*) > 0 \) and take \( n_\varepsilon \coloneqq \lceil \alpha^{-1} |\log \varepsilon| \rceil \). It  enables us to identify this large deviations limit as \( t_0^* \), using Proposition \ref{P4.7} (see also \cite{plachky1975theorem}).
Now, inequalities (\ref{low}) and (\ref{high})  together complete the proof.

\subsection{Proof of Theorem \ref{T4}}
In the following, let  $(a, b)\in \Gamma$ such that $F_{a,b}$ and and $\nu=\nu_\psi$ satisfies condition (H1).
Since $\nu $ is an ergodic Gibbs measure, $\nu$-almost every points $x\in \mathbb{I}$ is regular in the sense that the following limits
exists: \begin{equation}\label{dim}
\lim_{n \to \infty} \frac{1}{n} \log |df^n(x)| = \int \log |df| \d\nu, \quad 
\lim_{n \to \infty} \frac{1}{n} \log |dg^n_{a,b} (\phi^i)| = \lambda_{\nu,a,b}(\phi^i).
\end{equation}
Additionally, there are sequences of integers $n_1 < n_2 < \cdots$ and of reals $\varepsilon_1 > \varepsilon_2 > \cdots \to 0$ such that the symmetric $\varepsilon_k$-neighborhoods $V_{\varepsilon_k}(x)$ of $x$ satisfy
\[
f^{n_k}|_{V_{\varepsilon_k}} : V_{\varepsilon_k}(x) \to f^{n_k}(V_{\varepsilon_k}(x))
\]
is a diffeomorphism, and 
\[
\inf_k |f^{n_k}(V_{\varepsilon_k}(x))| > 0,
\]
and such that it is enough to evaluate
\[
\sigma_\nu^i(x,y) = \lim_{\varepsilon \to 0} \frac{1}{\log \varepsilon} \log \frac{\nu \times m(U_{\varepsilon}(x,y) \cap \mathbb{B}_i)}{\nu \times m(U_{\varepsilon}(x,y))}, 
\]
along the sequence $\varepsilon_k \to 0$, see \cite[Section ~5]{keller2014stability}.
Moreover, integers $n_k$ and reals $\varepsilon_k$ can be chosen that the following holds:
\begin{equation}\label{e1}
    \lim_{k \to \infty} \frac{\log \varepsilon_{k+1}}{\log \varepsilon_k} = \lim_{k \to \infty} \frac{\log |df^{n_{k+1}}(x)|}{\log |df^{n_k}(x)|} = \lim_{k \to \infty} \frac{n_{k+1}}{n_k} = 1.
\end{equation}
We obtain the next lemma using \cite[Lemma~3.8]{keller2017stability} with some modification.
\begin{lemma}
   Suppose that $y < \phi^1$. Then
\begin{align*}
 \liminf_{k \to \infty} \frac{1}{\log \varepsilon_k} \log \frac{\nu\big(V_{\varepsilon_k}(x) \cap \{\phi^* < y + \varepsilon_k \}\big)}{\nu(V_{\varepsilon_k}(x))} &\leq \sigma_\nu^1(x,y) \\&
\leq 
\limsup_{k \to \infty} \frac{1}{\log \varepsilon_k} \log \frac{\nu\big(V_{\varepsilon_k}(x) \cap \{\phi^* < y+ \varepsilon_k/2 \}\big)}{\nu(V_{\varepsilon_k}(x))}.
\end{align*}

\end{lemma}

Using the distortion bound (\ref{7}) for $f$ (note that $\nu$ is positive on open sets, by Lemma \ref{L00}), the following holds:
\[
\e^{-D} \leq \frac{\nu\big(V_{\varepsilon_k}(x) \cap \{\phi^* < y+ \varepsilon_k/2 \}\big)}{\nu(V_{\varepsilon_k}(x))} \cdot
\frac{\nu(f^n(V_{\varepsilon_k}(x)))}{\nu\big(f^{n_k}(V_{\varepsilon_k}(x) \cap \{\phi^* < y+ \varepsilon_k/2 \}\big)} \leq\e^D.
\]
Next, observe that 
\[
f^{n_k}(V_{\varepsilon_k}(x) \cap \{\phi^* < y+ \varepsilon_k/2 \}) = f^{n_k}(V_{\varepsilon_k}(x))\cap \{\phi^* <u_k\}
\]
where $u_k : f^{n_k}(V_{\varepsilon_k}(x)) \to \mathbb{I}$ satisfies $u_k(f^{n_k}(\tilde{x})) = g^{n_k}_{a,b,\tilde{x}}(y+ \frac{\varepsilon_k}{2})$ as in Lemma \ref{L3.3}. There, it is proved that
\[
\e^{-D} \cdot (u_k(f^{n_k}(x)) - \phi^0) \leq u_k(z) - \phi^0 \leq\e^D \cdot (u_k(f^{n_k}(x)) - \phi^0), \quad \forall z \in f^{n_k}(V_{\varepsilon_k}(x)).
\]
Thus, noting also that $\inf_k \nu(f^{n_k}(V_{\varepsilon_k}(x))) > 0$, so, we have for $y < \phi^1$:
\begin{align}\label{18}
 \sigma_\nu^1(x,y) &\leq \limsup_{k \to \infty} \frac{1}{\log \varepsilon_k} \log \frac{\nu(V_{\varepsilon_k}(x) \cap \{\phi^* < y + \varepsilon_k/2\})}{\nu(V_{\varepsilon_k}(x))},\nonumber \\&
\leq \limsup_{k \to \infty} \frac{1}{\log \varepsilon_k} \log \nu\big(f^{n_k}(V_{\varepsilon_k}(x)) \cap \{\phi^*-\phi^0<\e^{-D}(g^{n_k}_{a,b,x}(y + \frac{\varepsilon_k}{2}) - \phi^0)\big),\nonumber\\&  = 
\begin{cases} 
t_0^* \cdot \limsup_{k \to \infty} \frac{\log(g^{n_k}_{a,b,x}(y + \frac{\varepsilon_k}{2}) - \phi^0)}{\log \varepsilon_k /2}, & \text{if } \liminf_{k \to \infty} (g^{n_k}_{a,b,x}(y + \frac{\varepsilon_k}{2}) - \phi^0)= 0, \\
 0, & \text{otherwise.}
\end{cases}
\end{align}
We applied Theorem \ref{T3}  to obtain the equality. Similarly, the lower bound can be derived.
\begin{equation}\label{19}
\sigma^1_\nu(x,y) \geq
\begin{cases} 
t_1^* \cdot \liminf_{k \to \infty} \frac{\log(g^{n_k}_{a,b,x}(y + \varepsilon_k) - \phi^0)}{\log \varepsilon_k}, & \text{if } \limsup_{k \to \infty} (g^{n_k}_{a,b,x}(y + \varepsilon_k) - \phi^0)= 0, \\
0, & \text{if}\limsup_{k \to \infty} (g^{n_k}_{a,b,x}(y + \varepsilon_k) - \phi^0)> 0.
\end{cases}
\end{equation}
We observe that corresponding statements hold for $\sigma^0_\nu(x,y)$. 
Indeed, by (\ref{e1}) and from Birkhoff’s Ergodic Theorem for $\nu$-a.\,e.\@ $x$, we have:
\[
 -\lim_{k \to \infty} \frac{\log  \varepsilon_k}{n_k} = -\lim_{k \to \infty} \frac{\log \varepsilon_k/2}{n_k} =\lim_{n \to \infty} \frac{1}{n} S_n \log |df(f^j x)| = \int \log |df| \d\nu> 0,
\]
that we denote it by $\gamma$. We consider the following cases to prove Theorem \ref{T4}, ignoring a set of $x$'s of $\nu$-measure 0:
\\
\\
\noindent\emph{Case (1):}
For  $\phi^0< y < \phi^*(x)$,  we can assume that $\phi^0< \phi^*$ $\nu$-a.e., so that $\lambda_{a,b,\nu}(\varphi^0) < 0$ according to condition (H1), see also \cite[Proposition~ 1.6]{keller2017stability}.
As $y < \phi^*(x)$, there exists $k_0 = k_0(x,y)$ such that
\[
\lim_{n \to \infty} g^n_{a,b,x}(y + \varepsilon_{n_{ k_0}}) - \phi^0 = 0.
\]
Due to the monotonicity of the branches $g_{a,b,x}$, it follows that
\[
\lim_{k \to \infty} g^{n_k}(y + p \varepsilon_{n_k}) - \phi^0= 0, \quad \text{for each } p \in [0, 1].
\]
Thus, we can apply Lemma \ref{L3.2} and conclude that
\[
\lim_{k \to \infty} \frac{1}{n_k} \log\big(g^{n_k}_{a,b,x}(y + p \varepsilon_{n_k}) - \varphi^0\big) 
= \lim_{k \to \infty} \frac{1}{n_k} \log  dg^{n_k}_{a,b,x}(\phi^0) = \lambda_{a,b,\nu}(\phi^0) < 0,
\]
for $\nu$-a.e.\ $x$ and $y \in (\phi^0, \phi^*(x))$. Using (\ref{18}) and (\ref{19}), this implies for such $(x,y)$:
\[
\sigma^1_\nu(x,y) = t_0^* \cdot \lim_{k \to \infty} \frac{\log\big(g^{n_k}_{a,b,x}(y + \varepsilon_k) - \phi^0 \big)}{-\gamma n_k} 
= t_0^* \cdot \frac{-\lambda_{a,b,\nu}(\phi^0)}{\int \log |df| \d\nu} > 0.
\]
It follows that $\sigma^0_\nu(x,y) = 0$.
\\
\\
\noindent\emph{Case (2):}
As in the previous case, but now for  $\phi^*(x) < y < \phi^1$, one shows that 
\[
\sigma^0_\nu(x,y) = t_1^* \cdot \frac{-\log \lambda_{a,b,\nu}(\phi^1)}{\int \log |df| \d\nu} > 0
\]
and $\sigma^1_\nu(x,y) = 0$, for $\nu$-a.\,e.\ $x$ and $y \in (\phi^*(x), \phi^1)$.

\section{Multifractal analysis}\label{section5}
In this section, we perform a multifractal analysis for the Hausdorff dimension of the level sets of the stability index, and  prove Theorem \ref{T5} and Theorem \ref{T6}.

In the following, let   $(a, b)\in \Gamma$ such that  $F_{a,b}$ and $\nu=\nu_\psi$ satisfiy condition (H1). 

Take $r_0$ small enough so that $f$  restricted to any open ball with diameter $r_0$ has well-defined inverse branch.
  It is well known that the expanding Markov map $f$ has Markov partitions with arbitrarily small diameters. 
  Choose a Markov partition $\mathcal{R}=\{R_1, \dots, R_k\}$  such that its diameter is smaller than $r_0$. 
Let $n > 0$ and let $h$ be any branch of $f^{-n}$.  
Then, if  $\xi : \mathbb{I}\to \mathbb{R}$ is a Hölder continuous function, there exists $C_\xi > 0$ such that, whenever $x, \tilde{x} \in h(B_{r_0}(z))$, we have
\[
\big| S_n \xi(x) - S_n \xi(\tilde{x}) \big| \leq C_\xi,
\]
where $S_n\xi(x)=\sum_{i=0}^{n-1}\xi(f^i(x))$ (see \cite[Lemma~2.1]{walkden2017stability}).
In particular, we apply this result to $\log dg_{a,b}$ and $\log |df|$ (and writing $C_g$, $C_f$ in place of $C_{\log dg_{a,b}}$, $C_{\log |df|}$, respectively), we have that for all $n$ and all $x, \tilde{x}\in h(B_{r_0}(z))$
\begin{equation}\label{86}
C_g^{-1} \leq \frac{dg_{a,b,x}^n}{dg_{a,b,\tilde{x}}^n} \leq C_g, \quad C_f^{-1} \leq  \frac{\prod_{j=0}^{n-1} |df(f^j (x)|^{-1}}{\prod_{j=0}^{n-1} |df(f^j (\tilde{x})|^{-1}} \leq  C_f. 
\end{equation}
We let
\[
[i_0, \ldots, i_n] \coloneqq \{x \in \mathbb{I} \mid f^j x \in R_{i_j}, \text{ for } j = 0, 1, \ldots, n\}
\]
and call this a cylinder of rank $n$. If $x$ does not intersect the boundary of Markov partition, then we write $I_n(x)$ to be the unique cylinder of rank $n$ that contains $x$.

Let us consider the level sets $A_{0,\psi}(\sigma)$ and $A_{1,\psi}(\sigma)$ as defined in (\ref{M0}) and  (\ref{M1}), respectively.
For the functions 
\begin{equation}
    t \mapsto p_{i,\psi}(t)=P(\psi +t \log dg_{a,b}(\phi^i))=\sup_{\nu \in \mathcal{M}_f}\left(h_\nu(f) + \int \psi \d\nu +t \lambda_{\nu,a,b}(\phi^i)\right),
\end{equation}
$i=0,1$, we have, $p_{i,\psi}(0)=0$ , and $p_{i,\psi}^{\prime}(0)=\lambda_{ \nu_{\psi}, a, b}(\phi^i)$, for $a,b \in \Gamma$ (see e.\,g.\@\cite{keller2017stability}). By condition (H1),
 $\lambda_{\nu_\psi,a,b}(\phi^0)<0$ and $\lambda_{\nu_\psi,a,b}(\phi^1)<0$. Additionally, 
by the proof of Theorem \ref{T1} and Remark \ref{BI},  there are two $f$-invariant measure $\nu^i$, $i=0,1$, with $\lambda_{\nu^i, a, b}(\phi^i)>0$.
By this fact, the following holds:
\begin{equation*}
   \sup_{s>0} p_{0,\psi}(s) > 0 \quad  \text{and} \quad  \sup_{s>0} p_{1,\psi}(s) > 0.
\end{equation*}
Hence,  as we have seen before, the convex functions $p_{0,\psi}$ and $p_{1,\psi}$ have unique positive zeros $t_0^*$ and $t_1^*$, respectively.

Note that the stability index measures the degree of intermingledness near individual points, making it analogous to a local dimension. 
On the other hand, for ergodic equilibrium Gibbs measure $\nu_\psi$,  with $P(\psi)=0$, we apply the Birkhoff's ergodic theorem,  hence, for $\nu_\psi$-a.\,e.\@, points $x \in \mathbb{I}$, the following limit exist (see (\ref{dim})):
 $$\lim_{n \to \infty}\frac{1}{n}\log|df^n(x)|=\int_{\mathbb{I}}\log|df(x)|\d\nu_\psi, $$ and
 $$\ \lim_{n \to \infty}\frac{1}{n}\log |dg_{a,b,x}^n(\phi^i(x))|=\int_{\mathbb{I}}\log|dg_{a,b,x}(\phi^i(x))|\d\nu_\psi(x)=\lambda_{\nu_\psi,a,b}(\phi^i),$$
 for $i=0,1$. Therefore, 
\begin{equation}\label{A}
    A_{i,\nu_\psi}(\sigma)=\left\{x \in \mathbb{I}: \lim_{n \to \infty}\frac{t_i^* S_n \log |dg_{a,b,x}(\phi^i(x))|}{-S_n \log |df(x)|}=\sigma \right\}. 
\end{equation}
Thus,  multifractal analysis can be examined the Hausdorff dimension of the level sets of the stability index, as detailed in  \cite{pesin1997multifractal,walkden2017stability}.

Note that Pesin and Weiss \cite{pesin1997multifractal} perform a thorough multifractal analysis of equilibrium measures for Hölder continuous conformal expanding maps and Markov maps on an interval and for a wide range of Moran-like geometric constructions that meet a separation condition. 
A Moran cover provides the most efficient cover of $\mathbb{I}$ by cylinders of small diameter. An important property of Moran covers is the following: there exists $M > 0$ such that, for any $x \in \mathbb{I}$ and sufficiently small $r > 0$,  the number of sets in Moran cover $\mathcal{U}_r$ that have a non-empty intersection with $B_r(x)$ is bounded above by $M$. We refer to $M$ as the Moran multiplicity factor  (see also \cite{parry1990zeta}).

We give a multifractal analysis for Hausdorff dimension of the level sets of the stability index of  the attractor $A_1$. A similar argument can be applied for $A_0$.
 We recall the approached used by \cite{pesin1997multifractal}.
Let $\psi$ be the Hölder continuous potential as defined in (H1) with equilibrium measure $\nu=\nu_\psi$ with $P(\psi) = 0$.  First, note that the fiber map $g_{a,b}$ is strictly increasing, so we write $|\log dg_{a,b}|=\log dg_{a,b}$.
By Lemma \ref{L00}, $\nu$ is fully supported. Moreover, it is positive on each open ball.
Assume that $S(q)$ is uniquely determined by $P(-S(q) \log  |df| + q \log dg_{a,b}(\phi^1 ) )=0$  and  let $\nu_q$  denote the equilibrium state with potential $-S(q) \log  |df| + q \log dg_{a,b}(\phi^1 )$  (it is evident that, for each $q$, the function $-S(q) \log  |df| + q \log dg_{a,b}(\phi^1 )$  is Hölder continuous). Let 
\begin{equation*}
   \sigma(q)\coloneqq-dS(q)=\int \log dg_{a,b}(\phi^1 ) \d\nu_q / \int \log |df^{-1}|\d\nu_q. 
\end{equation*}
If $\log |df|$ is not cohomologous to $\log dg_{a,b}(\phi^1 )$  plus a constant, then $S(q)$ is a strictly convex analytic function and is the Legendre transform pair of the function 
$$L(\sigma)\coloneqq\dim_H \{x \in \mathbb{I}: \lim_{n \to \infty}S_n \log dg_{a,b,x}(\phi^1 )/S_n \log |df(x)|^{-1}=\sigma\},$$
so that $L(\sigma(q))=S(q)+q \sigma(q)$. Moreover, $L(\sigma(q))$ is defined on the interval $[\sigma(\infty), \sigma(- \infty)]$ . 

As in \cite{walkden2017stability} and \cite{pesin1997multifractal}, the following is a sketch proof to show that $L(\sigma(q))$ is the Legendre transform of $S(q)$.  
Indeed, it is known that $\nu_q(A_{1,\nu}(\sigma(q))) = 1$. Let $x \in A_{1,\nu}(\sigma(q))$, then 
$$\prod_{j=0}^{n-1}dg_{a,b}(\phi^1)(f^j(x)) \asymp \prod_{j=0}^{n-1}|df(f^j(x))|^{\sigma},$$
where $dg_{a,b}(\phi^1)(f^j(x))=dg_{a,b, f^j(x)}(\phi^1)$.
Assume $I_n(x)$ to be the unique cylinder of rank $n$ that contains $x$, and$I_n(x)$ has diameter approximately $r$.  Since $\nu_{q}$ is a Gibbs measure,  due to the properties of the Moran cover and using \cite[Lemma~2]{pesin1997multifractal}, the following holds:
\[\nu_q(I_n(x)) \asymp \prod_{j=0}^{n-1}|df(f^j(x))|^{-S(q)} dg_{a,b}(\phi^1)(f^j(x))^q \asymp \prod_{j=0}^{n-1}|df(f^j(x))|^{-(S(q))+q \sigma(q))} \asymp r^{S(q)+q \sigma(q)}\]
where the symbol ``$\asymp$'' means that the quotient of the right and left sides  are uniformly bounded away from zero and infinity. 
 This suggests that typical points in $A_{1,\nu}(\sigma(q))$ have a local dimension equal to $S(q)+q \sigma(q)$. More details will be given in the next subsection.

    \subsection{Proof of Theorems  \ref{T5} and \ref{T6}}
    
   In the following, we present the proof of Theorem \ref{T6}. A similar argument can be applied to establish the proof of Theorem \ref{T5}.
    To prove, we adopt the approach outlined in \cite[Section~6]{walkden2017stability}. We consider two cases.
   \\ 
   \\
  \noindent  \emph{Case (1):}  Consider the potential $\psi=-\log |df|$. Note that $P(-\log |df|)=0$ and $\nu_\psi=\nu_{\mathrm{ac}}$. By Theorem \ref{T3}, the exponent $t_1^*$ is defined by $P(-\log |df|+t_1^* \log dg_{a,b}(\phi^1))=0$. 
    Let $S(q)$ be defined by  
    \[P(-S(q) \log | df| + q t_1^*\log dg_{a,b}(\phi^1) ) = 0.\]
     To demonstrate that $S(q)$ is well defined, we consider 
  \[\Psi(q,r)\coloneqq -r \log | df| + q t_1^*\log dg_{a,b}(\phi^1).\]
    Then,  for the relevant Gibbs measure $\nu$,     
    \[\partial P(\Psi(q,r))/ \partial r=- \int \log |df|\d\nu\neq 0
    \]
        and the implicit function theorem then guarantees that  $S(q)$ is well defined. 
    Standard arguments based on the analyticity of the pressure demonstrate that $S(q)$ is analytic.
Note that by the definition of exponent $t_1^*$, we have $S(0) = 1$ and $S(1) = 1$. Let $\nu_q$ be the equilibrium state with potential $-S(q) \log  |df| + q t_1^*\log dg_{a,b}(\phi^1)$.
By differentiating $P(-S(q) \log  |df| + q t_1^*\log dg_{a,b}(\phi^1))=0$ with respect to $q$, we find that 
\[dS(q)= \frac{t_1^* \int \log dg_{a,b}(\phi^1)\d\nu_q}{\int \log|df|\d\nu_q}.\]
Next, we apply differentiation to the implicit function  $P(-S(q) \log  |df| + q t_1^*\log dg_{a,b}(\phi^1))=0$ twice with respect to $q$ and we apply a standard result from \cite{ruelle2004thermodynamic} to conclude that $d^2S(q) \geq 0$. In particular,  equality holds if and only if $t_1^*\log dg_{a,b}(\phi^1)$ and $dS(q) \log |df|$ are cohomologous up to a constant. 
Note that, by the proof of  Theorem \ref{T1} and Remark \ref{pos},  there exists an $f$-invariant measure $\nu^1$ with $\lambda_{\nu^1, a, b}(\phi^1)>0$ which implies that $\int t_1^* \log dg_{a,b}(\phi^1)\d\nu^1>0$, and also by the proof of Theorem \ref{T1}, $\int t_1^* dg_{a,b}(\phi^1)\d\nu_{\mathrm{ac}}<0$. However, since $\log |df| >0$, we have $\int dS(q)\log |df| \d\nu^1 $ and $\int dS(q)\log |df| \d\nu_{\mathrm{ac}}$ have the same sign, or are zero if $d S(q)=0$.
Hence, $d^2S(q) > 0$ and so $S(q)$ is a strictly convex function. 
As $S(q)$ is strictly convex, $S(0) = 1$ and $S(1) = 1$, there exists a unique
$q^* \in (0, 1)$ such that $dS(q^*)=0$.
If $q < q^*$ then $\int \log dg_{a,b}(\phi^1)\d\nu_q<0$.  Hence, by \cite[Proposition~1.6]{keller2017stability}, $\phi^*$ is defined $\nu_q$-a.\,e.\@ Let  $\sigma(q)=-dS(q)= -t_1^* \int \log dg_{a,b}(\phi^1)\d\nu_q / \int \log|df|\d\nu_q$. Then standard arguments from \cite{pesin1997multifractal} (sketched above) show that 
\[\dim_HA_{1,\nu_{\mathrm{ac}}}(\sigma(q))=S(q) + q \sigma(q), \]
the Legendre transform of $S(q)$, and that this is defined for $q< q^*$.
When $q = 0$ we have that $\nu_q = \nu_{\mathrm{ac}}$. Hence, 
\[
\dim_{H} \left\{x \in \mathbb{I}: \sigma_{\nu_{\mathrm{ac}}}(x,y)=\frac{t_1^* \int \log dg_{a,b}\d\nu_{\mathrm{ac}}}{\int \log |df|\d\nu_{\mathrm{ac}}}, \ \text{for} \ \text{all} \ y>\phi^*(x)\right\}=1.
\]
This completes the proof of Theorem \ref{T6} for the case that $\nu_\psi$ equals the SRB measure $\nu_{\mathrm{ac}}$.
\\
\\
\noindent\emph{Case (2):}  Assume that condition (H1)  holds for $F_{a,b}$ with $(a,b)\in \Gamma$ and   $\nu=\nu_\psi$ being the equilibrium state corresponding to the Hölder potential $\psi$ as required in condition (H1). 
As before, define $S(q)$ by the equation
\[
P\big(-S(q) \log |df| + q t_1^* \log dg_{a,b}(\phi^1)\big) = 0,
\]
and let $\nu_q$ denote the equilibrium state with potential 
\[
-S(q) \log |df| + q t_1^* \log dg_{a,b}(\phi^1).
\]
As discussed before, $S(q)$ is well defined, strictly convex, and satisfies
\[
dS(q) = t_1^* \frac{\int \log dg_{a,b}(\phi^1) \d\nu_q}{\int \log |df| \d\nu_q},
\]
where $d S(q)$ is the derivative of $S(q)$. Furthermore, it is noted that $S(0)$ satisfies the pressure equation   
\[
P(-S(0) \log |df|) = 0,
\]
we conclude that $S(0) = \dim_H \mathbb{I}=1$.
We first show that there exists $q \in \mathbb{R}$ such that 
\[
\int \log dg_{a,b}(\phi^1) \d\nu_q < 0.
\]
By the variational principle and definition of $S(q)$, the following lemma holds.

\begin{lemma}   For $(a,b)\in \Gamma$ and the Hölder continuous potential  $\psi$ we assume condition (H1) and let \( S(q) \) be defined as above. Then  

\[
S(q) = \sup \left\{ \frac{h(\nu) +q t_1^* \int \log dg_{a,b}(\phi^1) \d\nu }{\int \log  |df|\d\nu} : \nu\in \mathcal{M}_f  \right\},
\]
where \( h(\nu) \) denotes the entropy of \( f \) with respect to \( \nu \).  
\end{lemma}

As a consequence,  the following result follows  (cf. \cite{schmeling1999completeness, walkden2017stability}).

\begin{lemma}   For $(a,b)\in \Gamma$ and the Hölder continuous potential  $\psi$ we assume condition (H1). Then there exists \( q \in \mathbb{R} \) such that  
\[
\int \log dg_{a,b}(\phi^1) \d\nu_q < 0.
\]
\end{lemma}
\begin{proof}
Assume, for the sake of contradiction, that   
\[
\int \log dg_{a,b}(\phi^1) \d\nu_q \geq 0 \quad \text{for all } q \in \mathbb{R}.
\]
Then \( dS(q) \geq 0 \) for all \( q \in \mathbb{R} \). Since \( S \) is strictly convex, it follows that

\[
\alpha_0 \coloneqq \inf_{q \in \mathbb{R}} dS(q) = \lim_{q \to -\infty} dS(q) \geq 0.
\]
We demonstrate that this is not possible. Recall that if $\mu$ is any $f$-invariant probability measure, then 
\[
h_\mu(f) \leq h_{\text{top}}(f),
\]
where \( h_{\text{top}}(f) \) denotes the topological entropy of \( f \).  
By variational principle,  
\[
S(q) = \frac{h(\nu_q) + q t_1^* \int \log dg_{a,b}(\phi^1) \d\nu_q}{\int \log |df| \d\nu_q}.
\]
Let \( \epsilon > 0 \). Choose \( q < 0 \) such that  $\alpha_0 < dS(q) < \alpha_0 + \epsilon.$ Then,  $\alpha_0 < dS(q) + \epsilon$, so
\begin{align*}
 q \alpha_0 &> \frac{q t_1^* \int \log dg_{a,b}(\phi^1) \d\nu_q}{\int \log |df| \d\nu_q} + q \epsilon   \\
& \geq  S(q) - \frac{h(\nu_q)}{\int \log |df| \d\nu_q} + q \epsilon \\ & \geq \sup \left\{ \frac{h(\nu) + q t_1^* \int \log dg_{a,b}(\phi^1) \d\nu}{\int \log |df| \d\nu}\nu\in \mathcal{M}_f  \right\} - \frac{h_{\text{top}}(f)}{ \log ||df||} + q \epsilon \\&
\geq \sup \left\{ \frac{q t_1^* \int \log dg_{a,b}(\phi^1) \d\nu}{\int \log |df| \d\nu} :\nu\in \mathcal{M}_f \right\} - \frac{h_{\text{top}}(f)}{ \log ||df||} + q \epsilon.
\end{align*} 
 Dividing by $q$, letting $q \to - \infty$, and observing that $\epsilon>0$ is arbitrary, it follows that 
\[
\alpha_0 \leq \inf\left\{\frac{t_1^* \int \log dg_{a,b}(\phi^1) \d\nu }{ \int \log |df| \d\nu}: \nu\in \mathcal{M}_f \right \}.
\]
Let us take \( \nu = \nu_\psi \), then by (H1), we observe that \( \alpha_0 < 0 \). Hence, there exists \( \nu_q \) such that \(d S(q) < 0 \), a contradiction.

\end{proof}

If we repeat the above argument for \( q > 0 \) and let \( q\) tend to infinity, we conclude
\[
\sup_{q \in \mathbb{R}} dS(q) = \lim_{q \to \infty} dS(q) \geq \sup\left\{\frac{t_1^* \int \log dg_{a,b}(\phi^1) \d\nu }{ \int \log |df| \d\nu}: \nu\in \mathcal{M}_f \right \}.
\]
Taking \( \nu = \nu^1 \), as given in Remark \ref{pos}, we observe that \( dS(q) > 0 \) for all sufficiently large \( q \).
As \( S(q) \) is strictly convex, we have that \( dS(q) \) is increasing. Hence, there exists a unique \( q^* \in \mathbb{R} \) such that \( dS(q^*) = 0 \). Since  for  \( q < q^* \), 
\[
\int \log dg_{a,b}(\phi^1)\d\nu_q < 0,
\]
 we have that the invariant graph \( \phi^*\) is defined \( \nu_q \)-almost everywhere. Let 
\[
\sigma(q) =\frac{ -t_1^* \int \log dg_{a,b}(\phi^1) \d\nu_q }{ \int \log |df| \d\nu_q}.
\]
Using the standard arguments set out in \cite{pesin1997multifractal} (and briefly summarized above), it then follows that  
\[
\dim_H A_{1,\nu}(\sigma(q)) = S(q) + q\sigma(q)
\]
coincides with the Legendre transformation of \( S \) defined on $(-\infty, q^*)$ and evaluated at the point $\sigma(q)$.
Note that 
\begin{equation}
\dim_H  \left\{x \in \mathbb{I}: \sigma_\nu(x,y) = \frac{t_1^* \int  \log dg_{a,b}(\phi^1)\d\nu}{ \int  \log |df(x)| \d\nu} \ \text{ for all} \ y > \phi^*(x) \right\},
  \end{equation}
  is given by the unique \( q < q^*\)  for which 
\[
dS(q) =\frac{t_1^* \int  \log dg_{a,b}(\phi^1) \d\nu}{\int  \log |df(x)| \d\nu}.
\]

\end{proof}

\medskip
\printbibliography
\end{document}